\documentclass[11pt,reqno]{article}
\usepackage{amsmath, amssymb}
\usepackage{amsthm}
\usepackage[a4paper]{geometry}
\usepackage{fullpage} 
\usepackage[english]{babel}

\usepackage{tikz}
\usepackage{tikz-cd}
\usetikzlibrary{calc}
\usepackage{color, mathtools,epsfig, graphicx}

\usepackage[title]{appendix}

\usepackage{multirow}

\usepackage{hyperref}
\usepackage{aliascnt}

\newtheorem{theorem}{Theorem}

\theoremstyle{definition}

\newtheorem{remark}{Remark}

\addto\extrasenglish{

}

\newaliascnt{lemma}{theorem}  
\newtheorem{lemma}[lemma]{Lemma}  
\aliascntresetthe{lemma}  
\newaliascnt{definition}{theorem}  
\newtheorem{definition}[definition]{Definition}
\aliascntresetthe{definition}  
\newaliascnt{proposition}{theorem}  
\theoremstyle{theorem}
\newtheorem{proposition}[proposition]{Proposition}
\aliascntresetthe{proposition}  
\newcommand{\X}{\mathcal{X}}
\newcommand{\C}{\mathbb{C}}
\newcommand{\p}{\mathbb{P}}
\newcommand{\Tak}{\operatorname{Tak}}
\newcommand{\Ok}{\operatorname{Ok}}

\newcommand{\pain}[1]{\text{P}_{\mathrm{#1}}}

\begin{document}
\title{Takasaki's rational fourth Painlev\'e-Calogero system and geometric regularisability of algebro-Painlev\'e equations}
\author{Galina Filipuk\footnote{Institute of Mathematics, University of Warsaw, ul. Banacha 2, 02-097, 
Warsaw, Poland.
Email: filipuk@mimuw.edu.pl}, 
Alexander Stokes\footnote{Graduate School of Mathematical Sciences, The University of Tokyo, 3-8-1 Komaba Meguro-ku Tokyo 153--8914, Japan.
Email: stokes@ms.u-tokyo.ac.jp}}
\maketitle

\begin{abstract}
We study a Hamiltonian system without the Painlev\'e property and show that it admits a kind of regularisation on a bundle of rational surfaces with certain divisors removed, generalising Okamoto's spaces of initial conditions for the Painlev\'e differential equations.
The system in question was obtained by Takasaki as part of the Painlev\'e-Calogero correspondence and possesses the algebro-Painlevé property, being related by an algebraic transformation to the fourth Painlev\'e equation. 
We provide an atlas for the bundle of surfaces in which the system has a global Hamiltonian structure, with all Hamiltonian functions being polynomial in coordinates just as in the case of Okamoto's spaces.
We compare the surface associated with the Takasaki system with that of the fourth Painlev\'e equation, showing that they are related by a combination of blowdowns and a branched double cover, under which we lift the birational B\"acklund transformation symmetries of the fourth Painlev\'e equation to algebraic ones of the Takasaki system, including a discrete Painlev\'e equation.
We also discuss and provide more examples in support of the idea that there is a connection between the algebro-Painlev\'e property and similar notions of regularisability, in an analogous way to how regular initial value problems for the Painlev\'e equations everywhere on Okamoto's spaces are related to the Painlev\'e property.
\end{abstract}



\tableofcontents 

\section{Introduction}

It is well-known that solutions of a linear ordinary differential equation (ODE) can have only fixed singularities  in the complex plane, i.e. those whose locations are determined by the coefficients of the equation.  
For a nonlinear differential equation, solutions may have singularities whose locations depend on initial conditions and are, therefore, called movable.  
The Painlev\'e property for differential equations requires that all solutions are single-valued about all movable singularities, which is related to the ability to meaningfully define special functions in terms of the general solutions to such equations. 
This idea was pursued in the early 1900's by P. Painlev\'e and his student B. Gambier, as well as R. Fuchs, and they considered a large class of nonlinear second-order ODEs and classified those with this property into 50 classes. 
All equations could be solved in terms of known functions or in terms of solutions of what are now known as the six Painlev\'e equations, which we will denote by $\pain{J}$, $\mathrm{J} = \mathrm{I}, \dots, \mathrm{VI}$. 
These equations have remarkable properties and have found applications in many areas of mathematics and mathematical physics (see, for example, \cite{GLS} for further references). 

In a seminal paper \cite{Okamoto french}, K. Okamoto showed that each Painlev\'e equation is regularised on a complex manifold providing an extended phase space for an equivalent first-order system, constructed through a combination of compactification, blowups and removal of certain curves.
This associates a special kind of rational surface to the differential system, the geometry of which explains many of the properties of the equation.
Okamoto's paper was built on in a series of papers by Takano and collaborators \cite{ST97, MMT99, M97} proving that the manifolds for $\pain{II}$-$\pain{VI}$ essentially determine the equations (the case of $\pain{I}$ was shown later independently by Chiba \cite{chiba} and Iwasaki-Okada \cite{orbifold}), giving rise to the idea that geometry tells one everything about the Painlev\'e equations.
This idea, with its origins in Okamoto's work, has gone on to inform much of the recent study of Painlev\'e equations,
especially their discrete versions \cite{SAKAI2001}, via the associated geometric objects (see the important survey \cite{KNY} and references within).

The regularisation on Okamoto's space relies heavily on the Painlev\'e property, but it is natural to ask whether equations with weaker restrictions on singularities can still be associated to rational surfaces.
The main consideration to be made in addressing such a question is of what the appropriate relaxation of the requirement of regularisation is that that would allow such an association of surface to equation, and how it should correspond to the nature of singularities of solutions.
There have been some explorations in this direction, e.g. by Filipuk and Kecker \cite{FK} for a class of equations with the quasi-Painlev\'e property, 
but in this paper we consider a stronger requirement that solutions are not just locally but globally finitely branched in the sense of the algebro-Painlev\'e property, i.e. that all solutions are algebraic over the field of meromorphic functions on the universal cover of $\C \backslash F$, where $F$ is the finite set of fixed singularities of the equation.
We introduce a notion of global regularisability of a system of differential equations that reflects this special singularity structure in a way analogous to how regularisation of the Painlev\'e equations on Okamoto's spaces reflects the Painlev\'e property.

The first part of this paper is concerned with the regularisation of a non-autonomous rational Hamiltonian system obtained by Takasaki \cite{Takasaki} as part of his extension of the Painlev\'e-Calogero correspondence beyond the celebrated elliptic form of the sixth Painlev\'e equation due to Manin \cite{MANIN} (which built on the works of R. Fuchs \cite{FUCHS05} and Painlev\'e \cite{PAINLEVE} in which similar forms had already appeared). 
Though Takasaki's system does not possess the Painlev\'e property, we show that it is still possible to construct a family of rational surfaces which provide a manifold on which the system is everywhere either regular or regularisable by certain algebraic transformations.
This regularisability reflects the fact that the system is mapped by a rational (but not birational) transformation to the fourth Painlev\'e equation, which we describe on the level of the associated surfaces.
In particular this realises Okamoto's space for the fourth Painlev\'e equation as a quotient, under which we describe the lifts of symmetries of the Okamoto surface which provide B\"acklund transformations of $\pain{IV}$.
We also provide a global Hamiltonian structure of Takasaki's system on this manifold and compare this to that for $\pain{IV}$ on Okamoto's space due to Matano, Matumiya and Takano \cite{MMT99} (see also \cite{ST97}).

The study of Takasaki's system leads us to the main idea of the second part of the paper, namely that this is a special case of a kind of global regularisability on bundles of rational surfaces that is closely related to the algebro-Painlev\'e property.
We introduce such a notion of regularisability and demonstrate it in a suite of examples of algebro-Painlev\'e equations including an example obtained by Halburd and Kecker in \cite{HK} by Nevanlinna-theoretic methods, and show this regularisability can also function as a kind of algebro-Painlev\'e test. 

\subsection{Okamoto's space for the fourth Painlev\'e equation $\pain{IV}$} \label{section1.1}

For each of Painlev\'e equations Okamoto provided an equivalent first-order non-autonomous Hamiltonian system \cite{OkI}, and constructed an extended phase space on which the system is everywhere regular. 
We recall this construction in the case of the fourth Painlev\'e equation, which forms the motivation for  the present study and will be used in the statement of some results (for a detailed account including exposition of the relevant background material we refer the reader to  \cite{KNY}).
The fourth Painlev\'{e} equation in its usual scalar form is given by
\begin{equation} \label{P4scalar}
\text{P}_{\mathrm{IV}}:\quad \frac{d^2 \lambda}{dt^2} = \frac{1}{2\lambda}\left(\frac{d \lambda}{dt}\right)^2 + \frac{3}{2}\lambda^3 + 4t\lambda^2 + 2(t^2 - \alpha)\lambda + \frac{\beta}{\lambda},
\end{equation}
where $\alpha$ and $\beta$ are arbitrary complex parameters. 
The Okamoto Hamiltonian form of the equation \eqref{P4scalar} is given by
\begin{equation} \label{hamOkP4}
\begin{gathered}
\begin{aligned}
\frac{df}{dt} &= \frac{\partial H^{\Ok}}{\partial g} = 4 f g - f^2 - 2 t f - 2 a_1, \\
 \frac{dg}{dt} &= -\frac{\partial H^{\Ok}}{\partial f} = 2 f g -2 g^2 + 2 t g + a_2,
\end{aligned}
\\
H^{\Ok}(f,g,t) = 2 f g^2 - \left( f^2 + 2t f + 2 a_1 \right) g - a_2 f,
\end{gathered}
\end{equation}
where $a_1, a_2$ are complex parameters which we introduce here for convenience, corresponding to the root variables in the Sakai theory of rational surfaces associated with Painlev\'e equations \cite{SAKAI2001}. 
The reason that the first order system \eqref{hamOkP4} is regarded as a form of $\pain{IV}$ is that by eliminating $g$ we obtain the equation \eqref{P4scalar} for the function $f(t)$ with parameters 
\begin{equation}
\alpha = 1- a_1 - 2 a_2, \quad \beta = -2 a_1^2.
\end{equation}

As a non-autonomous system of two first-order ordinary differential equations, the phase space for \eqref{hamOkP4} is taken initially to be the product $\C^{2}_{f,g} \times \C_t$ (where subscripts indicate coordinates) as a trivial complex analytic fibre bundle over the independent variable space $\C_t$. To construct Okamoto's space we first compactify the fibres of the phase space to $\mathbb{P}^1 \times \mathbb{P}^1$ then perform a sequence of eight blowups of the fibre over $t$ in order to resolve the indeterminacies of the system, the details of which are deferred to \autoref{appendixA}.
We denote the points by $p_1,\dots, p_8$, with the exceptional divisor arising from the blowup of $p_i$ being denoted by $E_i$. 
Some of the points are infinitely near, i.e. there are cases when $p_{j+1}$ lies on the exceptional divisor $E_j$.
Through this procedure we obtain a rational surface $\X^{\Ok}_t$, which we give a schematic depiction of in \autoref{fig:surfaceokamoto}.

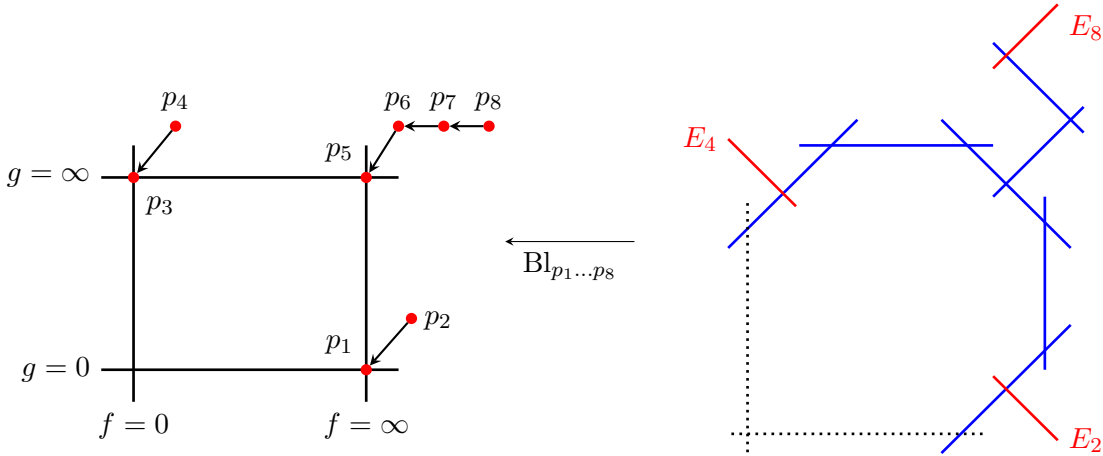
\begin{figure}[h]
	\begin{tikzpicture}[scale=.85,>=stealth,basept/.style={circle, draw=red!100, fill=red!100, thick, inner sep=0pt,minimum size=1.2mm}]
		\begin{scope}[xshift = -4cm]
			\draw [black, line width = 1pt] 	(4.1,2.5) 	-- (-0.5,2.5)	node [left]  {$g=\infty$} node[pos=0, right] {$$};
			\draw [black, line width = 1pt] 	(0,3) -- (0,-1)			node [below] {$f=0$}  node[pos=0, above, xshift=-7pt] {} ;
			\draw [black, line width = 1pt] 	(3.6,3) -- (3.6,-1)		node [below]  {$f=\infty$} node[pos=0, above, xshift=7pt] {};
			\draw [black, line width = 1pt] 	(4.1,-.5) 	-- (-0.5,-0.5)	node [left]  {$g=0$} node[pos=0, right] {$$};

			\node (p1) at (3.6,-.5) [basept,label={[xshift=-10pt, yshift = 0 pt] $p_{1}$}] {};
			\node (p2) at (4.3,0.3) [basept,label={[xshift=10pt, yshift = -10 pt] $p_{2}$}] {};
			\node (p3) at (0,2.5) [basept,label={[yshift=-20pt, xshift=+10pt] $p_{3}$}] {};
			\node (p4) at (0.65,3.3) [basept,label={[xshift=0pt, yshift = 0 pt] $p_{4}$}] {};
			\node (p5) at (3.6,2.5) [basept,label={[xshift=-10pt,yshift=0pt] $p_{5}$}] {};
			\node (p6) at (4.1,3.3) [basept,label={[xshift=0pt, yshift = 0 pt] $p_{6}$}] {};
			\node (p7) at (4.8,3.3) [basept,label={[xshift=0pt, yshift = 0 pt] $p_{7}$}] {};
			\node (p8) at (5.5,3.3) [basept,label={[xshift=0pt, yshift = 0 pt] $p_{8}$}] {};
			\draw [line width = 0.8pt, ->] (p2) -- (p1);
			\draw [line width = 0.8pt, ->] (p4) -- (p3);
			\draw [line width = 0.8pt, ->] (p6) -- (p5);
			\draw [line width = 0.8pt, ->] (p7) -- (p6);
			\draw [line width = 0.8pt, ->] (p8) -- (p7);
		\end{scope}
	
		\draw [->] (3.75,1.5)--(1.75,1.5) node[pos=0.5, below] {$\text{Bl}_{p_1\dots p_8}$};
	
		\begin{scope}[xshift = 6.5cm, yshift= .5cm]
			\draw [blue, line width = 1pt] 	(2.8,2.5) 	-- (-0.2,2.5)	node [pos = .5, below]  {
			} node[pos=0, right] {};
			\draw [blue, line width = 1pt] 	(3.6,1.7) -- (3.6,-1)		node [pos = .5, left]  {
			} node[pos=0, above, xshift=7pt] {};

			\draw [blue, line width = 1pt] 	(-1.3,0.9) 	-- (.7, 2.9)	 node [left] {} node[pos = 0, left] {
			};
			\draw [red, line width = 1pt] 	(-.25,1.55) 	-- (-1.3,2.6)	 node [left] {} node[pos=1, left] {$E_4$};
			
			\draw [blue, line width = 1pt] 	(4,-.3) node[left]{}	-- (2,-2.3)	 node [below left] {} node[below] {
			};

			\draw [red, line width = 1pt] 	(2.8,-1.1) 	-- (3.8,-2.1)	 node [left] {} node[pos=1, right] {$E_2$};
			
			\draw [blue, line width = 1pt]	(4,.9) -- (2,2.9) node[ pos=0, right] {
			};
			\draw [blue, line width = 1pt]	(2.8,1.7) -- (4.2,3.1) node [right] {
			} ;
			\draw [blue, line width = 1pt]	(4.2,2.7) -- (2.8,4.1) node [above] {
			};
			\draw [red, line width = 1pt]	(2.8,3.7) -- (3.8,4.7)  node [below right] {$E_8$};

			\draw [black, dotted, line width = 1pt]  (-1,1.6) -- (-1,-2.3);
			\draw [black, dotted, line width = 1pt]  (-1.25,-2) 	-- (2.7,-2);

%

		\end{scope}
	\end{tikzpicture}
	\caption{Okamoto surface $\X^{\Ok}_t$ for the fourth Painlev\'e equation}
	\label{fig:surfaceokamoto}
\end{figure}
Next it is necessary to remove the support $D^{\Ok}$ of the inaccessible divisors, on which the vector field defining the system is vertical with respect to the bundle structure over $\C_t$, which are the irreducible components of the unique effective anticanonical divisor of $\X^{\Ok}_t$ and are indicated in blue in \autoref{fig:surfaceokamoto}.
After this we denote the resulting complex analytic fibre bundle by
\begin{equation}
\begin{gathered}
\pi : E^{\Ok} \rightarrow \C_t, \\
\pi^{-1}(t) = E^{\Ok}_t = \X^{\Ok}_t - D^{\Ok}.
\end{gathered}
\end{equation}
To work with the differential system on $E^{\Ok}$ we use an atlas due to Matano, Matumiya and Takano \cite{MMT99}, the construction of which we also recall in \autoref{appendixA}. 
This consists of four charts: one coming from the original $f,g$ coordinates, as well as charts to cover the parts of $E_2, E_4, E_8$ away from the anticanonical divisor, indicated in red in \autoref{fig:surfaceokamoto}, on which the system is regular:
\begin{equation} \label{atlasOk}
E^{\Ok} = \C^3_{f,g,t} \cup \C^3_{x_1,y_1,t} \cup \C^3_{x_2,y_2,t} \cup \C^3_{x_3,y_3,t},
\end{equation}
with gluing defined by 
\begin{equation} \label{gluingOk}
\begin{gathered}
\frac{1}{f} = x_1, \quad g = x_1 \left( - a_2 + y_1 x_1\right), \\
f = x_2( a_1 + x_2 y_2), \quad \frac{1}{g} = x_2, \\
\frac{1}{f} = y_3, \quad \frac{1}{g} = \frac{y_3}{ \frac{1}{2} + y_3 \left( t + y_3 \left( a_1 + a_2 - 1 + x_3 y_3 \right) \right)} ,
\end{gathered}
\end{equation}
so the parts of $E_2$, $E_4$, and $E_8$ away from the inaccessible divisors are covered by charts $(x_1,y_1)$, $(x_2, y_2)$, and $(x_3,y_3)$ respectively.

The system \eqref{hamOkP4} extended to $E^{\Ok}$ via this atlas possesses a global holomorphic Hamiltonian structure.
As the system is non-autonomous and the gluing \eqref{gluingOk} includes transition functions which are $t$-dependent, such a structure is provided by a symplectic form $\omega^{\Ok}$ on the fibre over $t$ along with a collection of Hamiltonian functions, one in each chart. 
In particular the atlas is required to be such that the symplectic form written in coordinates is independent of $t$. That is, each transition function
\begin{equation*}
\varphi : \mathbb{C}^3 \ni (x,y,t) \mapsto \left(X(x,y,t),Y(x,y,t),t \right) \in \mathbb{C}^3,
\end{equation*}
with restriction to the fibre over $t$ denoted
\begin{equation*}
\varphi_t : \mathbb{C}^2 \ni (x,y) \mapsto \left(X(x,y;t),Y(x,y;t) \right) \in \mathbb{C}^2,
\end{equation*}
is required to satisfy 
\begin{equation} \label{symplecticcond}
F(x,y) dx \wedge dy = \varphi_t^* G(X,Y) dX \wedge dY,
\end{equation}
for some $t$-independent functions $F, G$, where $d$ denotes the exterior derivative on $\C^2$ and so $t$ is treated as a constant. 
With such an atlas, the following standard fact ensures that a Hamiltonian structure for the differential equation in one chart extends to the whole space $E^{\Ok}$.
\begin{lemma} \label{symplecticlemma}
With $\varphi$ as above satisfying the condition \eqref{symplecticcond}, then given $H(x,y,t)$ there exists $K(X,Y,t)$ (unique modulo functions of $t$) such that 
\begin{equation}
 F(x,y) dy \wedge dx - dH \wedge dt = \varphi^* \left( G(X,Y) dY \wedge dX - dK \wedge dt \right),
\end{equation}
where $d$ is the exterior derivative on $\C^3$ so $t$ is treated as a variable.
Further, the Hamiltonian system 
\begin{equation}
F(x,y) \frac{dx}{dt} = \frac{\partial H}{\partial y}, \qquad F(x,y) \frac{dy}{dt} = - \frac{\partial H}{\partial x}
\end{equation} 
is transformed under $\varphi$ to 
\begin{equation}
G(X,Y) \frac{dX}{dt} = \frac{\partial K}{\partial Y}, \qquad G(X,Y) \frac{dY}{dt} = - \frac{\partial K}{\partial X}.
\end{equation} 
\end{lemma}
Note that the Hamiltonians will not generally coincide under the transition functions - there may be correction terms arising from the $t$-dependence of the gluing. 
In the present case we take the symplectic form $\omega^{\Ok}$ on the fibre $E^{\Ok}_t$ given in coordinates by
\begin{equation} \label{canonicalcoordsOk}
\omega^{\Ok} = df \wedge dg = d x_1 \wedge d y_1 = d x_2 \wedge d y_2 = d x_3 \wedge d y_3,
\end{equation}
where the equalities are under the gluing \eqref{gluingOk} and $d$ here denotes the exterior derivative on the fibre $E^{\Ok}_t$, so $t$ is treated as a constant in the calculation.
By \autoref{symplecticlemma} the Hamiltonian function in the original $(f,g)$ coordinates provides a global Hamiltonian structure for the system on $E^{\Ok}$. 
In particular it turns out that the Hamiltonians in the other three charts are also polynomial in the coordinates:
\begin{equation}
\begin{aligned}
H^{\Ok}_0 &= H^{\Ok}_0(f,g,t) = 2 f g^2 - \left( f^2 + 2 t f + 2 a_1 \right) g - a_2 f, \\
H^{\Ok}_1 &= H^{\Ok}_1(x_1,y_1,t) =   2 \left( a_2 - x_1 y_1 \right) \left( t + y_1 (a_1 + a_2 - x_1 y_1) \right) - x_1, \\
H^{\Ok}_2 &= H^{\Ok}_2(x_2,y_2,t) =   2 y_2 - \left( a_1 +  x_2 y_2 \right) \left( 2 t + x_2 (a_1 + a_2 + x_2 y_2) \right) , \\
H^{\Ok}_3 &= H^{\Ok}_3(x_3,y_3,t) =  2 (a_2 - 1)(a_1+a_2 - 1) y_3 \\
&\qquad\qquad\qquad \qquad\qquad+ x_3 \left( 1 + 2 y_3 \left( t + y_3 \left( a_1+2 a_2 -2 + x_3 y_3 \right) \right)  \right).
\end{aligned}
\end{equation} 
It follows that the system on $E^{\Ok}$ is then everywhere regular, and each fibre $E^{\Ok}_t$ is called a \emph{space of initial conditions}. 
We refer to the total space $E^{\Ok}$ as the \emph{defining manifold}, following \cite{Takanodefiningmanifolds}, or as \emph{Okamoto's space} for $\pain{IV}$, while we refer to the compact surfaces $\X_t^{\Ok}$ (without inaccessible divisors removed) as the \emph{Okamoto surfaces}.


The defining manifold $E^{\Ok}$ is foliated by solutions of the extended system, in a way that is closely related to the Painlev\'e property of the equation. 
For each Painlev\'e equation $\pain{J}$,
 Okamoto provided an equivalent Hamiltonian system (with respect to the canonical symplectic form), which is polynomial in the two dependent variables, and analytic on the set $B_{\mathrm{J}} \subset \C$ given by the complement of the finite set of fixed singularities of $\pain{J}$ in the complex plane.
With phase space being the trivial complex analytic fibre bundle $( \C^2 \times B_{\mathrm{J}}, \pi_{\mathrm{J}}, B_{\mathrm{J}})$, where $\pi_{\mathrm{J}}$ is projection onto the second factor $B_{\mathrm{J}}$, solutions of the Hamiltonian system give a nonsingular foliation of $\C^2 \times B_{\mathrm{J}}$ into disjoint complex 1-dimensional leaves transverse to the fibres. 
However this foliation is not uniform, precisely because solutions may develop movable poles.
Along the same lines as outlined for $\pain{IV}$ above, Okamoto \cite{Okamoto french} constructed for each $\pain{J}$ a complex analytic fibre bundle $(E_{\mathrm{J}}, \pi_{\mathrm{J}} , B_{\mathrm{J}})$ that contains $( \C^2 \times B_{\mathrm{J}}, \pi_{\mathrm{J}}, B_{\mathrm{J}})$ as a fibre subspace and of which the extended Hamiltonian form induces a foliation into complex 1-dimensional solution curves such that
\begin{itemize}
\item Every leaf is transverse to the fibres and intersects the subspace $ \C^2 \times B_{\mathrm{J}}$,
\item The foliation is uniform: for each point $p_0 \in E_{\mathrm{J}}$, $\pi_{\mathrm{J}}(p_0) = t_0$, any path $\ell$ in $B_{\mathrm{J}}$ with starting point $t_0$ can be lifted to the leaf passing through $p_0$. 
In other words, any solution of the extended system of differential equations, with initial condition corresponding to $p_0$, can be holomorphically continued in $E_{\mathrm{J}}$ over $\ell$. 
\end{itemize}
The fact that Okamoto's Hamiltonian form of the fourth Painlev\'e equation induces a uniform foliation comes from the Painlev\'e property, namely that every solution of the Hamiltonian system for an initial condition in the initial phase space $\C^2 \times B_{\mathrm{J}}$ can be meromorphically continued along any path in $B_{\mathrm{J}}$.
However we emphasise that the regularisation on $E^{\Ok}$ alone is not sufficient to prove the Painlev\'e property and one requires, in addition to the regular initial value problems corresponding to singularities of solutions, certain auxiliary functions related to the global holomorphic Hamiltonian structure \cite{Takanodefiningmanifolds}, which are similar to those appearing in other proofs of the Painlev\'e property \cite{HL1, HL2, HL3, HL4, STEINMETZ, HUKUHARA, SHIMOMURA}.
Nevertheless the existence of a defining manifold on which an ODE system becomes everywhere regular is certainly intimately related to the Painlev\'e property, and this paper forms one part of a broader project extending this relation between regularisability and special singularity structures to wider classes of equations, e.g. to those with the quasi-Painlev\'e property \cite{FK}.

\subsection{Takasaki's rational Painlev\'e-Calogero system related to $\pain{IV}$}

We consider the second-order rational Painlev\'e-Calogero system related to $\pain{IV}$ obtained by Takasaki \cite{Takasaki}, which in this paper we call the Takasaki system, given explicitly by
\begin{equation} \label{takasakisystem}
\begin{gathered}
\begin{aligned}
\frac{dq}{dt} &= \frac{\partial H^{\Tak}}{\partial p} = p, \\
 \frac{dp}{dt} &= -\frac{\partial H^{\Tak}}{\partial q} =  \frac{8 \beta}{q^3} + (t^2 - \alpha) q + \frac{t}{2} q^3 + \frac{3}{64} q^5,
\end{aligned}
\\
\begin{aligned}
H^{\Tak}(q,p,t) &= \frac{p^2}{2} + V(q,p,t), \\
V(q,p,t) &= - \frac{1}{2} \left( \frac{q}{2} \right)^6 - 2 t \left( \frac{q}{2} \right)^4 - 2 (t^2 - \alpha) \left( \frac{q}{2} \right)^2 + \beta \left( \frac{q}{2} \right)^{-2},
\end{aligned}
\end{gathered}
\end{equation}
where 
$\alpha$, $\beta$ are again arbitrary complex parameters.

The potential $V$ is associated to the rank-one case of an extended Calogero system with rational potential \cite{INOZM, INOZ}, has an isomonodromic Lax pair \cite{BCR18}, and reduces in the $\beta=0$ case to an equation isolated by Halburd and Kecker \cite{HK} through Nevanlinna-theoretic methods, as we will see in \autoref{section4.1}. 
Takasaki obtained the system above as part of the Painlev\'e-Calogero correspondence, which began with the discovery of a so-called elliptic form of the sixth Painlev\'e equation by Manin \cite{MANIN}, building on the work of R. Fuchs \cite{FUCHS05} and Painlev\'e \cite{PAINLEVE}, which was recognised by Levin and Olshanetsky \cite{LO} as being related to a special case of one of Inozemtsev's generalisations \cite{INOZM, INOZ} of the Calogero system \cite{CALOGERO}. 
Takasaki extended the correspondence from $\pain{VI}$ to the other Painlev\'e equations - for a more detailed account of the historical progression of ideas leading to this we recommend the excellent paper \cite{TAKASAKIKYOTO}.

Eliminating $p$ from system \eqref{takasakisystem} one obtains the Takasaki scalar equation given by
\begin{equation} \label{takasakiscalar}
\frac{d^2 q}{dt^2}=\frac{3 q^5}{64}+\frac{t q^3}{2}+(t^2-\alpha)q+\frac{8\beta}{q^3},
\end{equation}
which is related to the fourth Painlev\'e equation as follows. If $(q,p)$ solve the Takasaki system (or $q$ solves the scalar equation \eqref{takasakiscalar}), then we obtain a solution $\lambda$ of the fourth Painlev\'e equation \eqref{P4scalar} with the same parameters $\alpha, \beta$ given by
\begin{equation}
\lambda = \left( \frac{q}{2} \right)^2.
\end{equation}
This gives a rational (but not birational) transformation that maps a solution $(q,p)$ of the Takasaki system to a solution $(f,g)$ of the Okamoto Hamiltonian form of the fourth Painlev\'e equation, which is given explicitly by \eqref{TaktoOktransformation} and will be studied in \autoref{section3}.

\subsection{Outline of the paper}

In \autoref{section2}, we introduce the notion of quadratic regularisability of a system of two ordinary differential equations and present a defining manifold $E^{\Tak}$ for the Takasaki system \eqref{takasakisystem} on which it is regular or regularisable in the sense of this definition. 
We then present a symplectic structure for the fibre over $t$, with respect to which we derive a global Hamiltonian structure for the Takasaki system on $E^{\Tak}$, in which the Hamiltonian in each chart is holomorphic.
The explicit details of the construction of the manifold are deferred to \autoref{appendixB}.

In \autoref{section3} we show how the algebraic transformation from the Takasaki system to the Okamoto Hamiltonian form of $\pain{IV}$ realises the space of initial conditions as a quotient. 
We also show how this map extended to the compact surfaces (without inaccessible divisors removed)
is a combination of blowdowns and a double cover, and clarify the relation under this map between the divisors of the rational 2-forms associated
with the Takasaki and Okamoto Hamiltonian systems.
We also examine symmetries of the Okamoto surface, which provide B\"acklund transformations of $\pain{IV}$, lifted to the Takasaki surface and in particular present a system of difference equations that is transformable algebraically but not birationally to a discrete Painlev\'e equation, which can be regarded as a discrete counterpart to the Takasaki system.

In \autoref{section4}, we define a class of regularising transformations that generalise the quadratic regularisability observed in the case of the Takasaki system. We discuss the ways in which this provides a notion of defining manifold for equations with the algebro-Painlev\'e property, with reference to a number of examples.

\section{Geometric regularisation of the Takasaki system} \label{section2}

The reason we choose the system \eqref{takasakisystem} among Takasaki's second-order Painlev\'e-Calogero systems as the object of study in this paper is that it is given in Hamiltonian form by a rational Hamiltonian function (others are elliptic, hyperbolic, or exponential).
Thus we can still recast it as a rational vector field on an appropriate phase space and attempt to resolve indeterminacies through blowups. 
The observations we make in doing so, the details of which will follow, lead us to introduce the following notion.
\begin{definition} \label{defquadreg}
Suppose a first-order system of ODEs is of the form
\begin{equation}
\frac{du}{dt} = R_1(t, u, v^2 ), \quad \frac{dv}{dt} = \frac{1}{v} R_2 (t,u, v^2),
\end{equation}
where $R_1(t,x,y)$ and $R_2(t,x,y)$ are rational functions of their arguments, both regular in $y$, with $R_2(t,x,y)$ nonzero  at $y=0$. 
Then we say that such a system is quadratic regularisable in the variable $v$. Introducing $h = v^2$, the resulting system is regular in $h$ at $h=0$:
\begin{equation}
u' = R_1( t,u, h ), \quad h' =  2R_2 (t,u, h).
\end{equation}
\end{definition}
After resolving the indeterminacies of the Takasaki system, we find that in the coordinates for the last exceptional divisors arising in the sequence of blowups, the system is precisely of the form above.
We also note that if after a sequence of blowups the system is quadratic regularisable in one chart $(u,v)$ for the last exceptional divisor arising in the sequence in the variable $v$, introduced in the canonical way according to the convention established in \autoref{appendixA}, then it is also quadratic regularisable in the other chart $(U,V)$, in the variable $V$, which follows from the relation
\begin{equation}
u v = V, \quad v = U V.
\end{equation}

\subsection{Defining manifold for the Takasaki system}

Similarly to the case of the Okamoto Hamiltonian form of $\pain{IV}$, we begin with the Takasaki system \eqref{takasakisystem} as a rational system on the trivial bundle over $\C_t$ with fibre $\C^2_{q,p}$, which we compactify to $\mathbb{P}^1 \times \mathbb{P}^1$.
We then perform a sequence of twenty blowups of the fibre over $t$, the details of which are provided in \autoref{appendixB}, denoting the centres of blowups by $q_i$ and exceptional divisors by $F_i$. 
Through this we obtain a rational surface $\X^{\Tak}$ (we sometimes add a subscript $t$ to emphasise the fibre), which is schematically represented in \autoref{fig:takasakisurface}.
Note we assume here that $\beta \neq 0$, since in the $\beta=0$ case some indeterminacies disappear and fewer blowups are required.

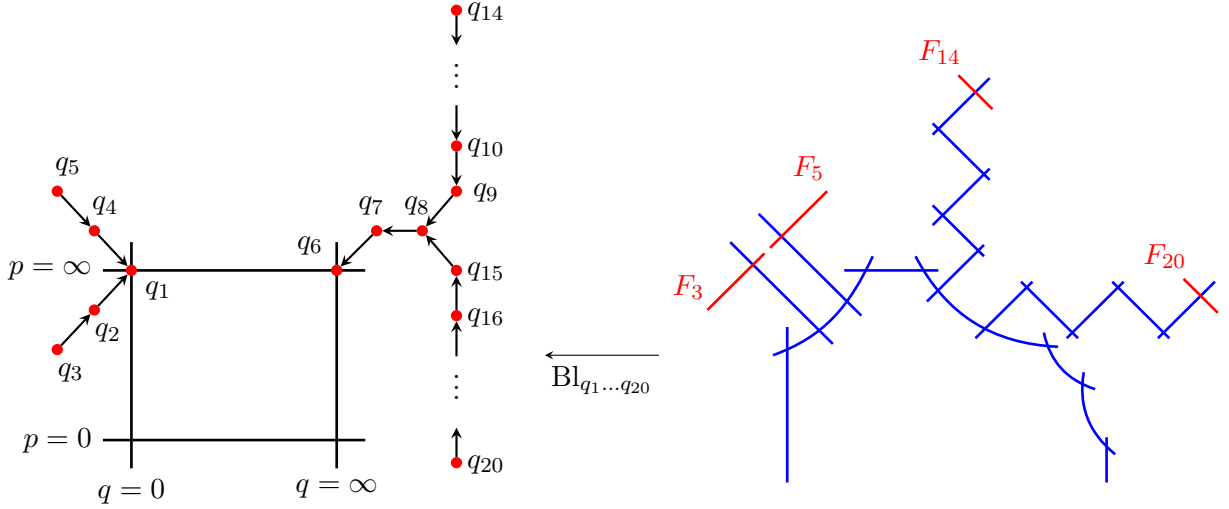
\begin{figure}[h] 
	\begin{tikzpicture}[scale=.75,>=stealth,basept/.style={circle, draw=red!100, fill=red!100, thick, inner sep=0pt,minimum size=1.2mm}]
	\draw [black, line width = 1pt] 	(4.1,2.5) 	-- (-0.5,2.5)	node [left]  {$p=\infty$} node[pos=0, right] {$$};
			\draw [black, line width = 1pt] 	(0,3) -- (0,-1)			node [below] {$q=0$}  node[pos=0, above, xshift=-7pt] {} ;
			\draw [black, line width = 1pt] 	(3.6,3) -- (3.6,-1)		node [below]  {$q=\infty$} node[pos=0, above, xshift=7pt] {};
			\draw [black, line width = 1pt] 	(4.1,-.5) 	-- (-0.5,-0.5)	node [left]  {$p=0$} node[pos=0, right] {$$};

			\node (q1) at (0,2.5) [basept,label={[yshift=-17pt, xshift=+10pt] $q_{1}$}] {};
			\node (q2) at (-0.65,1.8) [basept,label={[xshift=5pt, yshift = -18pt] $q_{2}$}] {};
			\node (q3) at (-1.3,1.1) [basept,label={[xshift=5pt, yshift = -18pt] $q_{3}$}] {};
			\node (q4) at (-0.65,3.2) [basept,label={[xshift=4pt, yshift = 0 pt] $q_{4}$}] {};
			\node (q5) at (-1.3,3.9) [basept,label={[xshift=4pt, yshift = 0 pt] $q_{5}$}] {};
			\node (q6) at (3.6,2.5) [basept,label={[xshift=-10pt,yshift=0pt] $q_{6}$}] {};
			\node (q7) at (4.3,3.2) [basept,label={[xshift=-2pt, yshift = 0 pt] $q_{7}$}] {};
			\node (q8) at (5.1,3.2) [basept,label={[xshift=-2pt, yshift = 0 pt] $q_{8}$}] {};
			\node (q9) at (5.7,3.9) [basept,label={[xshift=11pt, yshift = -10 pt] $q_{9}$}] {};
			\node (q10) at (5.7,4.7) [basept,label={[xshift=11pt, yshift = -10 pt] $q_{10}$}] {};
			\node (q11) at (5.7,5.6) [label={[xshift=0pt, yshift = -5 pt] $\vdots$}] {};
			\node (q12dots) at (5.7,6.3)  {};
			\node (q14) at (5.7,7.1) [basept,label={[xshift=11pt, yshift = -10 pt] $q_{14}$}] {};

			\node (q15) at (5.7,2.5) [basept,label={[xshift=11pt, yshift = -10 pt] $q_{15}$}] {};
			\node (q16) at (5.7,1.7) [basept,label={[xshift=11pt, yshift = -10 pt] $q_{16}$}] {};
			\node (q17) at (5.7,.8) [label={[xshift=0pt, yshift = -20 pt] $\vdots$}] {};
			\node (q18dots) at (5.7,-.1) {};
			\node (q20) at (5.7,-.9) [basept,label={[xshift=11pt, yshift = -10 pt] $q_{20}$}] {};
			\draw [line width = 0.8pt, ->] (q2) -- (q1);
			\draw [line width = 0.8pt, ->] (q3) -- (q2);
			\draw [line width = 0.8pt, ->] (q4) -- (q1);
			\draw [line width = 0.8pt, ->] (q5) -- (q4);
			\draw [line width = 0.8pt, ->] (q7) -- (q6);
			\draw [line width = 0.8pt, ->] (q8) -- (q7);
			\draw [line width = 0.8pt, ->] (q9) -- (q8);

			\draw [line width = 0.8pt, ->] (q10) -- (q9);
			\draw [line width = 0.8pt, ->] (q11) -- (q10);
			\draw [line width = 0.8pt, ->] (q14) -- (q12dots);

			\draw [line width = 0.8pt, ->] (q15) -- (q8);
			\draw [line width = 0.8pt, ->] (q16) -- (q15);
			\draw [line width = 0.8pt, ->] (q17) -- (q16);
			\draw [line width = 0.8pt, ->] (q20) -- (q18dots);
	
	\draw [->] (9.25,1)--(7.25,1) node[pos=0.5, below] {$\text{Bl}_{q_1\dots q_{20}}$};

	\begin{scope}[xshift=11.5cm]
			\draw [blue, line width = 1pt] 	(2.65,2.5) 	-- (1,2.5)	node [left]  {$$} node[pos=0, right] {$$};
			\draw [blue, line width = 1pt] 	(0,1.5) -- (0,-1.25)			node [below] {$$}  node[pos=0, above, xshift=-7pt] {} ;
			\draw [blue, line width = 1pt] 	(5.6,-.45) -- (5.6,-1.25)		node [below]  {$$} node[pos=0, above, xshift=7pt] {};

			\path [blue, line width = 1pt, out=20,in=-115]	(-.25,1) edge (1.5, 2.75) ;
			\draw [blue, line width = 1pt]	(.8,1.2) -- (-1, 3) ;
			\draw [red, line width = 1pt]		(-.4,2.8) -- (-1.4, 1.8)  node[pos=1, above, xshift=-7pt] {$F_{3}$} ;
			\draw [blue, line width = 1pt]	(1.3,1.7) -- (-.5, 3.5) ;
			\draw [red, line width = 1pt]		(-.3, 2.9) -- (.7, 3.9) node[pos=1, above, xshift=-7pt] {$F_{5}$} ;

			\path [blue, line width = 1pt, bend right]	(2.25,2.75) edge (4.75, 1.15) ;
			\path [blue, line width = 1pt, bend right]	(4.5,1.4) edge (5.4, .4) ;
			\path [blue, line width = 1pt, bend right]	(5.2,.7) edge (5.75, -.75) ;

			\draw [blue, line width = 1pt, bend right]	(2.45,1.95) -- (3.45, 2.95) ;
			\draw [blue, line width = 1pt]	(3.45,2.75) -- (2.55, 3.65) ;
			\draw [blue, line width = 1pt]	(2.55, 3.3)-- (3.55, 4.3) ;
			\draw [blue, line width = 1pt]	(3.55, 4.1)-- (2.55, 5.1) ;
			\draw [blue, line width = 1pt]	(2.55, 4.9) -- (3.55, 5.9) ;
			\draw [red, line width = 1pt]	(3, 5.95) -- (3.6, 5.35) node[pos=0, above, xshift=-7pt] {$F_{14}$} ;

			\draw [blue, line width = 1pt, bend right]	(3.3,1.3) -- (4.3, 2.3) ;
			\draw [blue, line width = 1pt]	(4.1, 2.3) -- (5.1, 1.3) ;
			\draw [blue, line width = 1pt]	(4.9, 1.3) -- (5.9, 2.3) ;
			\draw [blue, line width = 1pt]	(5.7, 2.3)  -- (6.7, 1.3) ;
			\draw [blue, line width = 1pt]	(6.5, 1.3)  -- (7.5, 2.3) ;
			\draw [red, line width = 1pt]	(6.95, 2.35)  -- (7.55, 1.75) node[pos=0, above, xshift=-7pt] {$F_{20}$} ;
		\end{scope}
	\end{tikzpicture}
		\caption{Surface $\X^{\Tak}$ for the Takasaki system, with inaccessible divisors in blue}
		\label{fig:takasakisurface}
\end{figure}

We next remove from the fibre certain divisors where the system is not regular or quadratic regularisable (indicated in blue in \autoref{fig:takasakisurface}), which we refer to as inaccessible by analogy with the Okamoto case and whose support we denote $D^{\Tak}$. 
After this we have the bundle $(E^{\Tak}, \pi, \C_t)$, the fibre of which is $E^{\Tak}_t = \X^{\Tak}_t  - D^{\Tak}$.
Through the blowups we obtain the following atlas for $E^{\Tak}$, provided by the initial variables $(q,p)$ away from $q=0$, as well as charts to cover affine neighbourhoods of the exceptional divisors $F_i$ for $i=3,5,14,20$ away from $D^{\Tak}$, as indicated in red in \autoref{fig:takasakisurface}.
\begin{equation} \label{atlasTak}
E^{\Tak} = \left( \C^3_{q,p,t} \backslash \{q=0\} \right) \cup \C^3_{u_3,v_3,t} \cup \C^3_{u_5,v_5,t} \cup \C^3_{u_{14},v_{14},t}\cup \C^3_{u_{20},v_{20},t},
\end{equation}
with gluing defined according to
\begin{equation} \label{takasakiatlas}
\begin{gathered}
q = v_3 ( 4 a_1 + u_3 v_3^2), \quad \frac{1}{p} = v_3, \\
q = v_5 ( -4 a_1 + u_5 v_5^2), \quad \frac{1}{p} = v_5, \\
\frac{1}{q} = v_{14}, \quad \frac{1}{p} = v_{14}^3\left(\frac{1}{8} + v_{14}^2\left(t + v_{14}^2 \left( 4( a_1+ 2 a_2 - 2) + u_{14} v_{14}^2\right) \right) \right), \\
\frac{1}{q} = v_{20}, \quad \frac{1}{p} = v_{20}^3\left(-\frac {1}{8} + v_{20}^2\left(-t + v_{20}^2 \left( -4 (a_1+ 2 a_2) + u_{20} v_{20}^2\right) \right) \right).
\end{gathered}
\end{equation}
For the construction of this atlas see \autoref{appendixB}, and note that we have introduced the parameters $a_1, a_2$ in the same way as in $\pain{IV}$, namely $\alpha = 1 - a_1 -2 a_2$, $\beta = -2 a_1^2$.
\begin{proposition}
The Takasaki system is everywhere either regular or quadratic regularisable on the bundle $E^{\Tak}$.
\end{proposition}

\begin{proof}
The main part of the result is the construction of the atlas in \autoref{appendixB}, after which it remains only to verify the assertion by direct calculation of the extension of the system \eqref{takasakisystem} to $E^{\Tak}$ using the relations \eqref{takasakiatlas} as changes of variables. 
For example in the chart $(u_{14},v_{14})$ we have 
\begin{equation}
\begin{aligned}
\frac{d u_{14}}{d t} &= 64(a_2 - 1)(a_1 + a_2 - 1) + 2 t u_{14} + 16(a_1 + 2 a_2 -2) u_{14} v_14^2 + 3 u_{14}^2 v_{14}^4,\\
\frac{d v_{14}}{d t} &= - \frac{1 + 8 t v_{14}^2 + 32(a_1 + 2 a_2 -2) v_{14}^2 + 8 u_{14} v_{14}^2}{v_{14} },
\end{aligned}
\end{equation}
so the system is quadratic regularisable in $v_{14}$ on the part of $F_{14}$ visible in the $(u_{14},v_{14})$ chart. 
The regularisability on the other exceptional divisors $F_{3}, F_{5}, F_{20}$ is verified similarly, and the system in the original variables is regular away from $q=0$. 
\end{proof}

\subsection{Global Hamiltonian structure}

The quadratic regularisability of the Takasaki system can also be seen in terms of a global Hamiltonian structure, in an analogous way to the global holomorphic Hamiltonian structure on Okamoto's space for $\pain{IV}$ as outlined \autoref{section1.1}. 
For this, we take the holomorphic symplectic form on the fibre $E^{\Tak}_t$ to be the extension of that with respect to which the Takasaki system \eqref{takasakisystem} is defined, namely $dq \wedge dp$. 
The appropriate atlas for $E^{\Tak}$ and collection of Hamiltonian functions are provided by the following.
\begin{theorem} \label{theoremTakasakiSymplectic}
The Takasaki system on $E^{\Tak}$ has a global Hamiltonian structure with respect to the symplectic form on the fibre $E_t^{\Tak}$ extended from $dq \wedge dp$ in the original variables, in which all Hamiltonian functions are polynomial in coordinates.
\end{theorem}

\begin{proof}
We take as an atlas for the bundle 
\begin{equation}\label{atlasTakpoly}
E^{\Tak} = \C^3_{x_1,y_1, t} \cup \C^3_{x_2,y_2, t} \cup \C^3_{x_3,y_3, t} \cup \C^3_{x_4,y_4, t},
\end{equation}
with gluing defined by 
\begin{equation} \label{takasakiatlassymplectic}
\begin{gathered}
\frac{1}{x_1} = \frac{y_2^2}{8 a_1 + x_2 y_2^2} ,\quad y_1 = y_2, \\
x_1 =  - y_3^2 ( 8 a_2 + x_3 y_3^2) ,\quad \frac{1}{y_1} = y_3, \\
\frac{1}{x_1} = \frac{4 y_4^2}{1 + 8 t y^2 +32(a_1 + a_2 -1) y_4^4 - 4 x_4 y_4^6}, \quad \frac{1}{y_1} = y_4,
\end{gathered}
\end{equation}
where notation has been recycled from the $\pain{IV}$ case and the charts $(x_i,y_i)$ are not to be confused with those from the symplectic atlas for $E^{\Ok}$.
The original variables $q,p$ are related to these coordinates by 
\begin{equation}
q = y_1, \quad p = x_1 y_1 - \frac{4 a_1}{y_1} - t y_1 - \frac{y_1^3}{8},
\end{equation}
and the transition functions $(x_i,y_i) \mapsto (x_j,y_j)$ defined by \eqref{takasakiatlassymplectic} can be verified to be biholomorphisms on the overlaps of coordinate patches by direct calculation. 
In particular the parts of the exceptional divisors $F_3, F_5, F_{14}$, and $F_{20}$ away from $D^{\Tak}$ are given by $y_2=0$, $y_1=0$, $y_4=0$, and $y_3=0$ respectively.
We take the symplectic form $\omega^{\Tak} = dq \wedge dp$ extended to the fibre $E_t^{\Tak}$, which is given in each chart by 
\begin{equation} \label{omegatakasakixy}
\omega^{\Tak} = y_i dx_i \wedge y_i,
\end{equation}
with respect to which the system is Hamiltonian: 
\begin{equation} \label{omegaxyi}
\frac{dx_i}{dt} = \frac{1}{y_i} \frac{\partial H^{\Tak}_i}{\partial y_i}, \quad \frac{dy_i}{dt} = -\frac{1}{y_i} \frac{\partial H^{\Tak}_i}{\partial x_i},
\end{equation}
with Hamiltonians $H^{\Tak}_i$ being polynomial in $x_i, y_i^2, t$, explicitly given (modulo functions of $t$) by 
\begin{equation}
\begin{aligned}
H^{\Tak}_1(x_1, y_1,t) &= 4 a_1 x_1 + \left( a_2 + t x_1 - \frac{x_1^2}{2} \right) y_1^2 + \frac{1}{8} x_1 y_1^4,\\
H^{\Tak}_2(x_2, y_2, t) &= - 4 a_1 x_2 + \left(a_1 + a_2 + t x_2 - \frac{x_2^2}{2} \right) y_2^2 + \frac{1}{8} x_2 y_2^4, \\
H^{\Tak}_3(x_3,y_3,t) &=  - \frac{1}{8} x_3 - \left( 32 a_2(a_1+a_2) + t x_3 \right) y_3^2 - 4 (a_1+2a_2) x_3 y_3^4 - \frac{1}{2} x_3^2 y_3^6, \\
 H^{\Tak}_4(x_4,y_4,t) &=  \frac{1}{8} x_4 -  \left( 32 (a_1+a_2-1)(a_2-1) -  t x_4 \right) y_4^2 + 4(a_1+2a_2-2)x_4 y_4^4 - \frac{1}{2} x_4^2 y_4^6.
\end{aligned}
\end{equation}
\end{proof}
\begin{remark}
In particular the fact that each Hamiltonian is polynomial in coordinates with only even powers of $y_i$ ensures that the system in the chart $(x_i, y_i)$ is quadratic regularisable in the variable $y_i$, and regular in $x_i$ everywhere.
This is analogous to how a polynomial Hamiltonian structure in an atlas providing canonical coordinates for the symplectic form guarantees regular initial value problems, as is the case for Okamoto's spaces for the Painlev\'e equations.
\end{remark}
\section{The map from Takasaki to Okamoto surfaces} \label{section3}

The following algebraic transformation was presented by Takasaki \cite{Takasaki}, and can be reproduced using the map from the scalar equation \eqref{takasakiscalar} to $\pain{IV}$.
If $(q,p)$ solves the system \eqref{takasakisystem} with $\alpha = 1 - a_1 -2 a_2$, $\beta = -2 a_1^2$, then the transformation 
\begin{equation} \label{TaktoOktransformation}
\left\{
\begin{aligned}
f(q,p)&= \left(\frac{q}{2} \right)^2, \\
 g(q,p) &= \frac{t}{2} + \frac{2 a_1 }{q^2} + \frac{p}{2 q} + \frac{q}{16},
\end{aligned}
\right.
 \end{equation}
gives a solution $(f,g)$ to the Okamoto Hamiltonian form \eqref{hamOkP4} of $\pain{IV}$. Further, this transformation is, up to scaling, canonical:
\begin{equation} \label{dqdp=dfdg}
dq \wedge dp = 4 df \wedge dg,
\end{equation}
where the equality is under the above transformation and $d$ is the exterior derivative on $\C^2$ so $t$ is treated as constant.
The above transformation defines a rational, but not birational, map
\begin{equation} \label{mapphi}
\varphi : \X^{\Tak} \rightarrow \X^{\Ok},
\end{equation}
which we study in this section.

\subsection{The rational map between open surfaces}  

We first consider the restriction of the map \eqref{mapphi} to the fibres of the defining manifolds $E^{\Tak}$, $E^{\Ok}$, so the surfaces with inaccessible divisors removed:
\begin{equation} \label{mapopensurfaces}
\varphi : 
\X^{\Tak} - D^{\Tak} \rightarrow \X^{\Ok} - D^{\Ok} 
.
\end{equation}
\begin{proposition}
The rational map \eqref{mapopensurfaces} is a morphism (i.e. has no indeterminacies) and has empty critical locus (i.e. does not blow down any curves and all fibres are finite). 
Further, its ramification locus consists of the parts away from $D^{\Ok}$ of the exceptional divisors $E_2, E_4, E_8$ and the proper transform of $\left\{f = 0\right\}$.
\end{proposition}

\begin{proof}
The proof is a direct calculation in coordinates, in which we use $\bar{u}_i, \bar{v}_i$ to denote charts for the Okamoto surface $\X^{\Ok}$ as introduced in \autoref{appendixA} to distinguish them from $u_i, v_i$ for the Takasaki surface $\X^{\Tak}$ as introduced in \autoref{appendixB}. 
It is straightforward to first verify that $\varphi$ has no indeterminacies on $\X^{\Tak}$, after which the absence of any critical locus as well as the ramification are checked through the following calculations:
\begin{itemize}
\item The part of $\X^{\Tak} - D^{\Tak}$ visible in the $(q,p)$-chart (so on which in particular $q \neq 0$) is mapped 
to the $(f,g)$-chart with $f \neq 0$, 
and every point $(f,g)$ with $f\neq 0$ has two distinct preimages in the $(q,p)$-chart away from $q=0$.
\item The part of the exceptional divisor $F_5$ on $\X^{\Tak} - D^{\Tak}$ is mapped 
to the part of the proper transform of $\{f=0\}$ on $\X^{\Ok}$ away from $D^{\Ok}$:
\begin{equation*}
\left. (u_5,v_5)\right |_{v_5 = 0} \mapsto (f,g) = \left( 0, \frac{u_5}{32 a_1^2}+ \frac{t}{2} \right),
\end{equation*}
and every point on this proper transform has preimage under $\varphi$ being a single point of $F_5$.
\item The part of the exceptional divisor $F_3$ on $\X^{\Tak} - D^{\Tak}$ is mapped 
to the part of $E_4$ on $\X^{\Ok}$ away from $D^{\Ok}$:
\begin{equation*}
\left. (u_3,v_3)\right |_{v_3 = 0} \mapsto (\bar{u}_4, \bar{v}_4) = \left( \frac{u_3}{32 a_1} +  \frac{t a_1}{2} , 0 \right),
\end{equation*}
and every point on $E_4$ away from $D^{\Ok}$ has preimage being a single point of $F_5$.
\end{itemize}
Similarly $F_{14}$ and $F_{20}$ are mapped 
to $E_8$ and $E_2$, respectively.
\end{proof}

\subsection{The rational map between compact surfaces}

We now investigate the rational map $\varphi$ between the compact surfaces and the relation of the rational 2-forms providing the symplectic structures for the Takasaki and Okamoto systems.
In the case of $\X^{\Ok}$, as with surfaces associated with the other Painlev\'e equations, the configuration of the inaccessible divisors forming $D^{\Ok}$ plays a defining role and they form the irreducible components of an effective anticanonical divisor of canonical type \cite{SAKAI2001}.
In particular they are all of self-intersection $-2$, with intersection configuration encoded in a Dynkin diagram of an affine root system, and give the pole divisor of the symplectic form used to define the Hamiltonian structure of $\pain{IV}$. 
In this section we establish the role played by the components of $D^{\Tak}$ in the symplectic structure with respect to which the Takasaki Hamiltonian system is defined.

In order to investigate the map $\varphi$ we first note that it is possible to perform a kind of minimisation of the Takasaki surface $\X^{\Tak}$ which does not affect the fibres $E^{\Tak}_t$ of the defining manifold, since it is possible to contract inaccessible exceptional curves contained in $D^{\Tak}$.
This is done through a sequence of four blowdowns, the details of which are provided in \autoref{appendixB}. 
This gives a surface $\X^{\Tak}_{m}$, 
which we call the minimal surface for the Takasaki system. 
We give a schematic representation of this in \autoref{mintakasakisurface},
where the components of the image of $D^{\Tak}$ under the blowdowns are labeled $C_i$, $i=1,\dots,15$ as in \autoref{appendixB} and indicated in a combination of blue and magenta for reasons that will be explained below.
\begin{figure}[h]
\centering
	\begin{tikzpicture}[scale=1,>=stealth,basept/.style={circle, draw=red!100, fill=red!100, thick, inner sep=0pt,minimum size=1.2mm},
	elt/.style={circle,draw=black!100, fill=blue!75, thick, inner sep=0pt,minimum size=2mm}, 
				mag/.style={circle, draw=black!100, fill=magenta!100, thick, inner sep=0pt,minimum size=2mm}  ]
			\draw [blue, line width = 1pt] 	(-.5,2.5) 	-- (4.5,2.5)	node [left]  {$$} node[pos=0.5, xshift=3pt, yshift=8pt] {$C_4$};
			\draw [blue, line width = 1pt] 	(0,3) -- (0,-2.5)			node [below] {$$}  node[pos=.6, above, xshift=-8pt] {$C_3$} ;
			\draw [blue, line width = 1pt] 	(4,3) -- (4,-2.5)		node [below]  {$$} node[pos=.6, above, xshift=8pt] {$C_5$};

			\draw [magenta, line width = 1pt] 	(.4,1) 	-- (-2,1)	node [left]  {$$} node[pos=0, right] {$C_2$};
			\draw [red, line width = 1pt] 	(-1.6,1.7) 	-- (-1.6,.3)	node [left]  {$$} node[pos=0, left] {$F_5$};

			\draw [magenta, line width = 1pt] 	(.4,-1) 	-- (-2,-1)	node [left]  {$$} node[pos=0, right] {$C_1$};
			\draw [red, line width = 1pt] 	(-1.6,-.3) 	-- (-1.6,-1.7)	node [left]  {$$} node[pos=1, left] {$F_3$};
			\draw [magenta, line width = 1pt] 	(3.6,1) 	-- (6,1)	node [left]  {$$} node[pos=0, left,xshift=2pt ] {$C_6$};
			\draw [blue, line width = 1pt] 	(5.6,0.6) 	-- (5.6,2.6) 	node [left]  {$$} node[pos=0.5, left] {$C_7$};
			\draw [magenta, line width = 1pt] 	(5.2,2.2) 	-- (7.2,2.2) 	node [left]  {$$} node[pos=0.5, above] {$C_8$};
			\draw [blue, line width = 1pt] 	(6.8,0.6) 	-- (6.8,2.6)	node [left]  {$$} node[pos=0.5, right] {$C_9$};
			\draw [magenta, line width = 1pt] 	(6.4,1) 	-- (8.4,1) 	node [left]  {$$} node[pos=1, below right] {$C_{10}$};
			\draw [red, line width = 1pt] 	(8,1.7) 	-- (8,.3) 	node [left]  {$$} node[pos=0, right] {$F_{14}$};
			\draw [magenta, line width = 1pt] 	(3.6,-1) 	-- (6,-1)	node [left]  {$$} node[pos=0, left, xshift=4pt ]{$C_{11}$};
			\draw [blue, line width = 1pt] 	(5.6,-.6) 	-- (5.6,-2.6) 	node [left]  {$$} node[pos=0.5, left] {$C_{12}$};
			\draw [magenta, line width = 1pt] 	(5.2,-2.2) 	-- (7.2,-2.2) 	node [left]  {$$} node[pos=0.5, below] {$C_{13}$};
			\draw [blue, line width = 1pt] 	(6.8,-0.6) 	-- (6.8,-2.6)	node [left]  {$$} node[pos=0.5, right] {$C_{14}$};
			\draw [magenta, line width = 1pt] 	(6.4,-1) 	-- (8.4,-1) 	node [left]  {$$} node[pos=1, above right] {$C_{15}$};
			\draw [red, line width = 1pt] 	(8,-1.7) 	-- (8,-.3) 	node [left]  {$$} node[pos=0, right] {$F_{20}$};

		\end{tikzpicture}
		\caption{Surface for the Takasaki Hamiltonian system after minimisation}
\label{mintakasakisurface}
\end{figure}
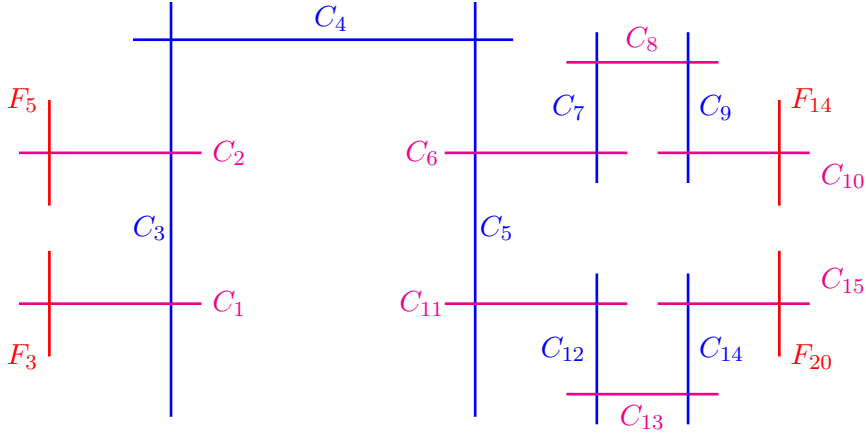

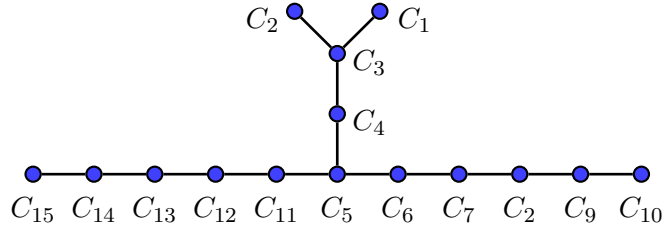
\begin{figure}[h] 
\centering
	\begin{tikzpicture}[scale=.8, 
				elt/.style={circle,draw=black!100, fill=blue!75, thick, inner sep=0pt,minimum size=2mm},
				blk/.style={circle,draw=black!100, fill=black!75, thick, inner sep=0pt,minimum size=2mm}, 
				mag/.style={circle, draw=black!100, fill=magenta!100, thick, inner sep=0pt,minimum size=2mm}, 
				red/.style={circle, draw=black!100, fill=red!100, thick, inner sep=0pt,minimum size=2mm} ]
				
		\path 	(-4,0) 	node 	(d1) [elt, label={[xshift=-0pt, yshift = -25 pt] $C_{14}$} ] {}
		        		(-2,0) 	node 	(d2) [elt, label={[xshift=-0pt, yshift = -25 pt] $C_{12}$} ] {}
		      		( 0,0) 	node  	(d3) [elt, label={[xshift=0pt, yshift = -25 pt] $C_5$} ] {}
		       		( 2,0) 	node  	(d4) [elt, label={[xshift=0pt, yshift = -25 pt] $C_7$} ] {}
		        		( 4,0) 	node 	(d5) [elt, label={[xshift=0pt, yshift = -25 pt] $C_9$} ] {}
		        		( 0,1) 	node 	(d6) [elt, label={[xshift=12pt, yshift = -15 pt] $C_4$} ] {}
		        		( 0,2) 	node 	(d0) [elt, label={[xshift=12pt, yshift = -15 pt] $C_3$} ] {};
		\path 	(.7,2.7)     node         (m1) [elt, label={[xshift=13pt, yshift = -15 pt] $C_1$} ] {}
				(-.7, 2.7)	node         (m2) [elt, label={[xshift=-12pt, yshift = -15 pt] $C_2$} ] {};
		\path 	(1,0) 		node		(m3)  [elt, label={[xshift=0pt, yshift = -25 pt] $C_6$} ] {}
				(3,0) 		node		(m4)  [elt, label={[xshift=0pt, yshift = -25 pt] $C_2$} ] {}
				(5,0) 		node		(m8)  [elt, label={[xshift=0pt, yshift = -25 pt] $C_{10}$} ] {};
		\path 	(-1,0) 	node		(m5)  [elt, label={[xshift=0pt, yshift = -25 pt] $C_{11}$} ] {}
				(-3,0) 	node		(m6)  [elt, label={[xshift=0pt, yshift = -25 pt] $C_{13}$} ] {}
				(-5,0) 	node		(m7)  [elt, label={[xshift=0pt, yshift = -25 pt] $C_{15}$} ] {};

		\draw [black,line width=1pt ] (m7) -- (d1) -- (m6) -- (d2) -- (m5) -- (d3);
		\draw [black,line width=1pt ] (d3) -- (m3)  -- (d4) -- (m4) -- (d5) -- (m8) ;
		\draw [black,line width=1pt ] (d3) -- (d6) -- (d0) ;		
		\draw [black,line width=1pt ] (d0) -- (m1) ;		
		\draw [black,line width=1pt ] (d0) -- (m2) ;		
	\end{tikzpicture}
\caption{Intersection configuration of inaccessible divisors for the Takasaki system}
\label{dynkintakasaki}
\end{figure}

In particular the self-intersection numbers of these curves are given by
\begin{equation}
(C_4)^2 = -4, \qquad (C_i)^2 = -2 ~~\text{    otherwise,}
\end{equation}
and their intersection configuration is given in \autoref{dynkintakasaki}.
It is natural to ask how this configuration of the curves $C_i$ transforms into that of the inaccessible divisors on $\X^{\Ok}$, which is encoded in the Dynkin diagram of type $E_6^{(1)}$.
From now on we take $\varphi$ as a map from the minimal surface $\X^{\Tak}_{m}$ to $\X^{\Ok}$. The fact that the transformation \eqref{TaktoOktransformation} is canonical in the sense of equation \eqref{dqdp=dfdg} means that the rational 2-forms $\omega^{\Tak}$ and $\omega^{\Ok}$ extended from $dq \wedge dp$ and $4 df\wedge dg$ respectively are related by
\begin{equation}
\varphi^* \omega^{\Ok} = \omega^{\Tak}.
\end{equation} 
While the pole divisor of $\omega^{\Ok}$ provides the effective anticanonical divisor of $\X^{\Ok}$, it turns out that the anticanonical divisor class of $\X^{\Tak}_m$ is not effective.
\begin{proposition}
The divisor of $\omega^{\Tak}$ on $\X^{\Tak}_m$ is given by
\begin{equation}
\begin{aligned}
- \operatorname{div} \omega^{\Tak} &=  C_3 + 2 C_4 + 5 C_5 + 4 C_6 + 3 C_7 + 2 C_8 \\
&\qquad + C_9+ 4C_{11} + 3 C_{12} + 2 C_{13} + C_{14} \\
&\qquad \qquad  - F_3 - F_5 - F_{14} - F_{20}.
\end{aligned}  
\end{equation}
\end{proposition}
\begin{proof}
This is obtained by a standard computation rewriting $\omega^{\Tak}$ in charts to cover the exceptional divisors. 
For example the curve $C_6$ corresponds to $F_9 - F_{10}$ and is given in the chart $(u_9, v_9)$ by $v_9=0$. The 2-form in this chart is computed directly to be
\begin{equation}
\omega^{\Tak} = \frac{dv_9 \wedge du_9}{v_9^4 (8 + u_9 v_9)^2},
\end{equation}
so we find the term $4 C_6$ in $- \operatorname{div} \omega^{\Tak}$. 
On the other hand the last exceptional divisors $F_3, F_5, F_{14}, F_{20}$ appear with negative coefficients in $- \operatorname{div} \omega^{\Tak}$ because the symplectic form has zeroes along them, for example in the chart $(u_3, v_3)$on $F_3$ we have
\begin{equation}
\omega^{\Tak} = v_3 dv_3 \wedge du_3,
\end{equation}
which is to be expected given the form of $\omega^{\Tak}$ in the atlas provided in \autoref{theoremTakasakiSymplectic}.
\end{proof}
The unique effective anticanonical divisor of $\X^{\Ok}$ is given by
\begin{equation}
 - \operatorname{div} \omega^{\Ok} = D_0 + D_1 + 2 D_2 + 3 D_3 + 2  D_4 + D_5 + 2 D_6 
 , 
\end{equation}
where the irreducible components $D_i$ are the same as those giving $D^{\Ok}$ as in \autoref{appendixA}.
Direct calculation in charts shows that $\varphi$ has no indeterminacies on $\X^{\Tak}_m$, but that it blows down some of the curves $C_i$, specifically those indicated in magenta on \autoref{mintakasakisurface}.
We apply 8 extra blowups to the images of these curves on $\X^{\Ok}$, to obtain what we call the extended Okamoto surface $\tilde{\X}^{\Ok}$, for which the projection is denoted
\begin{equation*}
\rho : \tilde{\X}^{\Ok} \longrightarrow \X^{\Ok}.
\end{equation*}
We use the notation $z_1, \dots, z_8$ for the extra points on $\X^{\Ok}$ to be blown up, with the corresponding exceptional divisors being $L_1, \dots, L_8$. 
The locations of the extra points are indicated in \autoref{fig:extendedOkamotosurface}, and given explicitly in coordinates in \autoref{appendixA}.

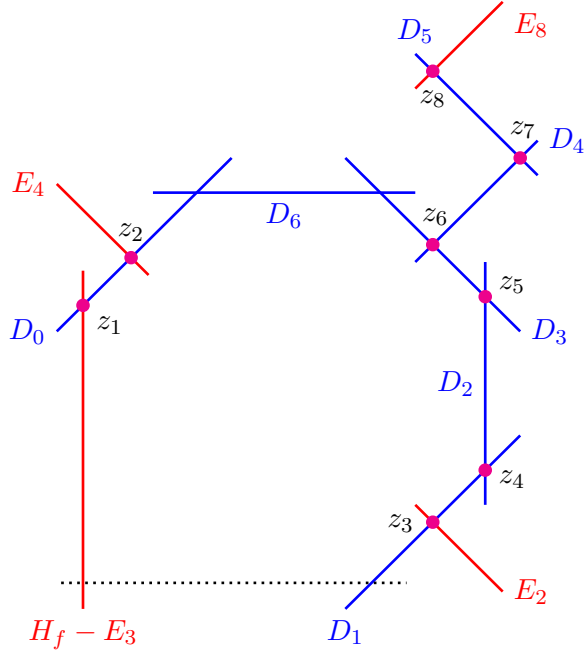
\begin{figure}[h]
\centering
	\begin{tikzpicture}[scale=1.15,>=stealth, basept/.style={circle, draw=red!100, fill=red!100, thick, inner sep=0pt,minimum size=1.2mm}, extrabasept/.style={circle, draw=magenta!100, fill=magenta!100, thick, inner sep=0pt,minimum size=1.5mm}]
			\draw [blue, line width = 1pt] 	(2.8,2.5) 	-- (-0.2,2.5)	node [pos = .5, below]  {$D_6$} node[pos=0, right] {};
			\draw [blue, line width = 1pt] 	(3.6,1.7) -- (3.6,-1.1)		node [pos = .5, left]  {$D_2$} node[pos=0, above, xshift=7pt] {};

			\draw [blue, line width = 1pt] 	(-1.3,0.9) 	-- (.7, 2.9)	 node [left] {} node[pos = 0, left] {$D_0$};
			\draw [red, line width = 1pt] 	(-.25,1.55) 	-- (-1.3,2.6)	 node [left] {} node[pos=1, left] {$E_4$};
			
			\draw [blue, line width = 1pt] 	(4,-.3) node[left]{}	-- (2,-2.3)	 node [below left] {} node[below] {$D_1$};

			\draw [red, line width = 1pt] 	(2.8,-1.1) 	-- (3.8,-2.1)	 node [left] {} node[pos=1, right] {$E_2$};
			
			\draw [blue, line width = 1pt]	(4,.9) -- (2,2.9) node[ pos=0, right] {$D_3$};
			\draw [blue, line width = 1pt]	(2.8,1.7) -- (4.2,3.1) node [right] {$D_4$} ;
			\draw [blue, line width = 1pt]	(4.2,2.7) -- (2.8,4.1) node [above] {$D_5$};
			\draw [red, line width = 1pt]	(2.8,3.7) -- (3.8,4.7)  node [below right] {$E_8$};

			\draw [red, line width = 1pt]  (-1,1.6) -- (-1,-2.3) node[below] {$H_f - E_3$};
			\draw [black, dotted, line width = 1pt]  (-1.25,-2) 	-- (2.7,-2);
			\node (z1) at (-1, 1.2) [extrabasept,label={[yshift=-17pt, xshift=+10pt] $z_{1}$}] {};
			\node (z2) at (-.45, 1.75) [extrabasept,label={[yshift=0pt, xshift=0pt] $z_{2}$}] {};
			\node (z5) at (3, 1.9) [extrabasept,label={[yshift=1pt, xshift=1pt] $z_{6}$}] {};
			\node (z6) at (4, 2.9) [extrabasept,label={[yshift=1pt, xshift=1pt] $z_{7}$}] {};
			\node (z7) at (3, 3.9) [extrabasept,label={[yshift=-20pt, xshift=0pt] $z_{8}$}] {};
			\node (z8) at (3.6, 1.3) [extrabasept,label={[yshift=-5pt, xshift=+10pt] $z_{5}$}] {};
			\node (z9) at (3.6, -.7) [extrabasept,label={[yshift=-13pt, xshift=+10pt] $z_{4}$}] {};
			\node (z10) at (3, -1.3) [extrabasept,label={[yshift=-10pt, xshift=-12pt] $z_{3}$}] {};

	\end{tikzpicture}
	\caption{Extra blowups for the extended Okamoto surface}
	\label{fig:extendedOkamotosurface}
\end{figure}

Denote the proper transform under $\rho$ of $D_i$ by $\tilde{D}_i \in \operatorname{Div}(\tilde{\X}^{\Ok})$, so 
\begin{equation}
\begin{aligned}
\tilde{D}_0 &= E_3 - E_4 - L_1 -L_2,\\
\tilde{D}_1 &= E_1 - E_2 - L_3 -L_4, \\
\tilde{D}_2 &= H_f - E_1 - E_5 - L_4 - L_5, \\
\tilde{D}_3 &= E_5 - E_6 - L_5 - L_6,
\end{aligned}
\qquad 
\begin{aligned}
\tilde{D}_4 &= E_6 - E_7 - L_6 - L_7, \\
\tilde{D}_5 &= E_7 - E_8 - L_7 - L_8, \\
\tilde{D}_6 &= H_g - E_3 - E_5. 
\end{aligned}
\end{equation}
The pole divisor of the rational 2-form $\tilde{\omega}^{\Ok} = \rho^* \omega^{\Ok}$ is given in terms of these and the extra exceptional divisors $L_i$ by 
\begin{equation}
- \operatorname{div}  \tilde{\omega}^{\Ok} = \tilde{D}_0 + \tilde{D}_1+ 2 \tilde{D}_2 + 3 \tilde{D}_3 + 2 \tilde{D}_4 + \tilde{D}_5 +2 \tilde{D}_6 + 2 L_4 + 4 L_5 + 4 L_6 + 2 L_7 
.
\end{equation}

\begin{theorem} \label{theorem7}
After minimisation of the Takasaki surface, the map given by the transformation \eqref{TaktoOktransformation} decomposes as 
\begin{equation*}
\rho \circ \tilde{\varphi} : \X_m^{\Tak} \longrightarrow \X^{\Ok},
\end{equation*}
where the factor
\begin{equation*}
\tilde{\varphi} : \X^{\Tak}_m \longrightarrow \tilde{\X}^{\Ok},
\end{equation*}
is a rational morphism and has no critical locus.
%
%
Then $\tilde{\varphi} : \X_m^{\Tak} \rightarrow \tilde{\X}^{\Ok}$ maps the curves $C_i$ according to
\begin{equation}
\tilde{\varphi} : 
\left\{
\begin{aligned}
C_1 &\mapsto L_2, 			& C_2 &\mapsto L_1, 		& C_3 & \mapsto \tilde{D}_0, 	& C_4 & \mapsto  \tilde{D}_6, 	& C_5 & \mapsto \tilde{D}_3,	 \\
C_6 & \mapsto L_6, 			& C_7 & \mapsto \tilde{D}_4, 	& C_8 & \mapsto L_7, 		& C_9 & \mapsto \tilde{D}_5, 	& C_{10} & \mapsto L_8,		\\
C_{11} & \mapsto L_5,		& C_{12} & \mapsto \tilde{D}_2, & C_{13} & \mapsto L_4,		& C_{14} & \mapsto \tilde{D}_1,	& C_{15} & \mapsto L_3,
\end{aligned}
\right.
\end{equation} 
and the last exceptional divisors $F_3, F_5, F_{14}, F_{20}$ by
\begin{equation}
\tilde{\varphi} ~:~ F_{3} \mapsto H_f - E_3 - L_1, \quad  F_5 \mapsto E_4 - L_2, 	\quad F_{14}  \mapsto E_8 - L_8,	\quad F_{20}  \mapsto E_2- L_3.
\end{equation}

\end{theorem}
\begin{proof}
First, we verify that there are no indeterminacies of $\tilde{\varphi}$ on $\X^{\Tak}_m$ by direct calculation in charts. 
The fact that the critical locus is empty is verified by computing the map in charts for $\X^{\Tak}_m$ as well as the new coordinates $(r_i, s_i)$ and $(R_i, S_i)$ covering the exceptional divisors $L_i$ on $\tilde{\X}^{\Ok}$. 
For example the fact that $C_{11}$, corresponding to $F_{15} - F_{16}$, is mapped to $L_5$ can be verified by rewriting the transformation \eqref{TaktoOktransformation} in charts $(u_{15},v_{15})$ and $(r_5,s_5)$, which gives
\begin{equation}
\begin{aligned}
r_5 &= \frac{ \left( - 64 v_{15} \left( t + 4 a_1 v_{15}^2 \right) + u_{15} \left(1 + 8 t v_{15}^2 + 32 a_1 v_{15}^4 \right) \right)^2}{64 (8 - u_{15} v_{15})^2},\\
s_5 &= \frac{ 16 v_{15} \left( 8 - u_{15}v_{15} \right)}{64 v_{15} \left( t + 4 a_1 v_{15}^2 \right) - u_{15} \left( 1 + 8 t v_{15}^2 + 32 a_1 v_{15}^4 \right)},
\end{aligned}
\end{equation}
into which substitution of the local equation $v_{15}=0$ of $C_{11}$ leads to 
\begin{equation}
\left. (u_{15},v_{15})\right |_{v_{15} = 0} \mapsto (r_5, s_5) = \left( \frac{u_{15}^2}{4096}, 0 \right),
\end{equation}
which given that $s_5=0$ is a local equation of $L_5$ allows us to deduce that $C_{11}$ is mapped surjectively onto $L_5$.
The rest of the calculations are similar, albeit with more complicated rational functions.
\end{proof}

\begin{proposition}
The action of $\tilde{\varphi}$ by pullback on components of the image $\tilde{\varphi}(D^{\Tak})$ is given by:
\begin{equation}
\tilde{\varphi}^* : 
\left\{
\begin{aligned}
\tilde{D}_0 	&	\mapsto 2 C_3,   
& \tilde{D}_1 	&	 \mapsto 2C_{14},  	
&\tilde{D}_2 	&	\mapsto  2 C_{12}, 
&\tilde{D}_3 	&	\mapsto 2 C_5, \\
\tilde{D}_4 	&	\mapsto 2 C_7, 
&\tilde{D}_5 	&	\mapsto 2 C_9, 
&\tilde{D}_6 	&	\mapsto C_4,
& &  \\
\vspace{+3em}
 L_1 		&	\mapsto C_2, 
& L_2 		&	\mapsto C_1,  
& L_3 		&	\mapsto C_{15},   
& L_4 		&	\mapsto C_{13},  \\
 L_5 		&	\mapsto C_{11}, 
& L_6		&	\mapsto C_6,  
& L_7 		&	\mapsto C_8, 
& L_8		&	\mapsto C_{10}.
\end{aligned}
\right.
\end{equation}
\end{proposition}
\begin{proof}
This follows from direct calculation using local equations of the curves in charts, along the same lines as in the proof of \autoref{theorem7}. For example, to see that $\tilde{\varphi}^*(\tilde{D}_0)=2C_3$, we take the chart $(u_3,v_3)$ for $\tilde{\X}^{\Ok}$ as in \autoref{appendixA}, and relabel $u,v$ by $\tilde{u}, \tilde{v}$ to distinguish them from the charts for the minimised Takasaki surface $\X_{m}^{\Tak}$. 
In this chart $(\tilde{u}_3,\tilde{v}_3)$, the local equation of $\tilde{D}_0$ is $\tilde{v}_3 = 0$.
On $\X_{m}^{\Tak}$ we take the chart $(u_1,v_1)$ as defined in \autoref{appendixB}, in which $C_3$ (the image of $F_1 - F_2 - F_4$ under the minimisation $\rho$) has local equation $v_1 = 0$. 
Rewriting the mapping $\tilde{\varphi}$ in these charts using their definitions in \autoref{appendixA} and \autoref{appendixB}, we see that 
\begin{equation}
\begin{aligned}
\tilde{u}_3 &= \frac{1}{64} \left( 32 a_1+ u_1 \left( 8 + 8 t u_1 v_1^2 + u_1^3 v_1^4 \right) \right),\\
\tilde{v}_3 &= \frac{16 u_1^2 v_1^2}{32 a_1+ u_1 \left( 8 + 8 t u_1 v_1^2 + u_1^3 v_1^4 \right)},
\end{aligned}
\end{equation}
so the divisor $\tilde{D}_0$ with local equation $\tilde{v}_3=0$ is pulled back to twice the divisor defined by $v_1=0$, i.e. $2 C_3$. 

On the other hand, the cases in which divisors are pulled back without increase in multiplicity can be seen by similar calculations. 
For example to see that $\tilde{\varphi}^*(L_2)=C_1$, consider the chart $(r_2,s_2)$ for $\tilde{\X}^{\Ok}$ in which $L_2$ is given by $s_2=0$, and the chart $(u_2, v_2)$ for $\X_{m}^{\Tak}$, in which $C_1$ (which we recall is the image of $F_2-F_3$ under $\rho$) is given by $v_2=0$. 
Calculating by direct substitution into the defining equations \eqref{TaktoOktransformation} of the mapping we obtain a relation of the form
$s_2 = v_2 \left(\frac{u_2}{8} + v_2 P(u_2, v_2)\right)$, where $P$ is polynomial in $u_2$, $v_2$, so in this case the divisor $L_2$ with local equation $s_2=0$ is pulled back to $C_1$, counted with multiplicity one as opposed to two. 
The rest of the calculations are similar.
\end{proof}
So in particular we see that eight of the $-2$ curves $C_i$ are mapped to exceptional curves $L_j$ of the first kind, which can be understood through the formula
\begin{equation}
\tilde{\varphi}^* C \cdot \tilde{\varphi}^* C = ( \operatorname{deg} \tilde{\varphi} ) C \cdot C = 2 C \cdot C,
\end{equation}
so in particular
\begin{equation}
(\tilde{\varphi}^* L_i)^2 = 2 (L_i)^2 = -2.
\end{equation}
These exceptional curves $L_j$ are then contracted by $\rho$, leading to the configuration of $-2$ curves providing $D^{\Ok}$. 
We give a graphical depiction of the transformation of the configuration of curves $C_i$ to that of $D_j$ in \autoref{dynkinsforphi}, as well as the exceptional divisors $F_3, F_5, F_{14}, F_{20}$. 

Of the curves $C_i$ on $\X^{\Tak}_m$, it is $C_{3}, C_{4}, C_{5}, C_{7}, C_{9}, C_{12},$ and $C_{14}$ (indicated in blue on \autoref{mintakasakisurface}) that are mapped to components of the anticanonical divisor of $\X^{\Ok}$, while the rest of the $C_i$ (indicated in magenta on \autoref{mintakasakisurface}) are mapped first to $-1$ curves under $\tilde{\varphi}$ then are contracted to points by $\rho$.
The exceptional divisors $F_3$, $F_5$, $F_{14}$, $F_{20}$ which are not contained in $D^{\Tak}$ first become $-2$ curves under $\tilde{\varphi}$ then become the divisors $E_4$, $H_f-E_3$, $E_8$, $E_2$ respectively on $\X^{\Ok}$.

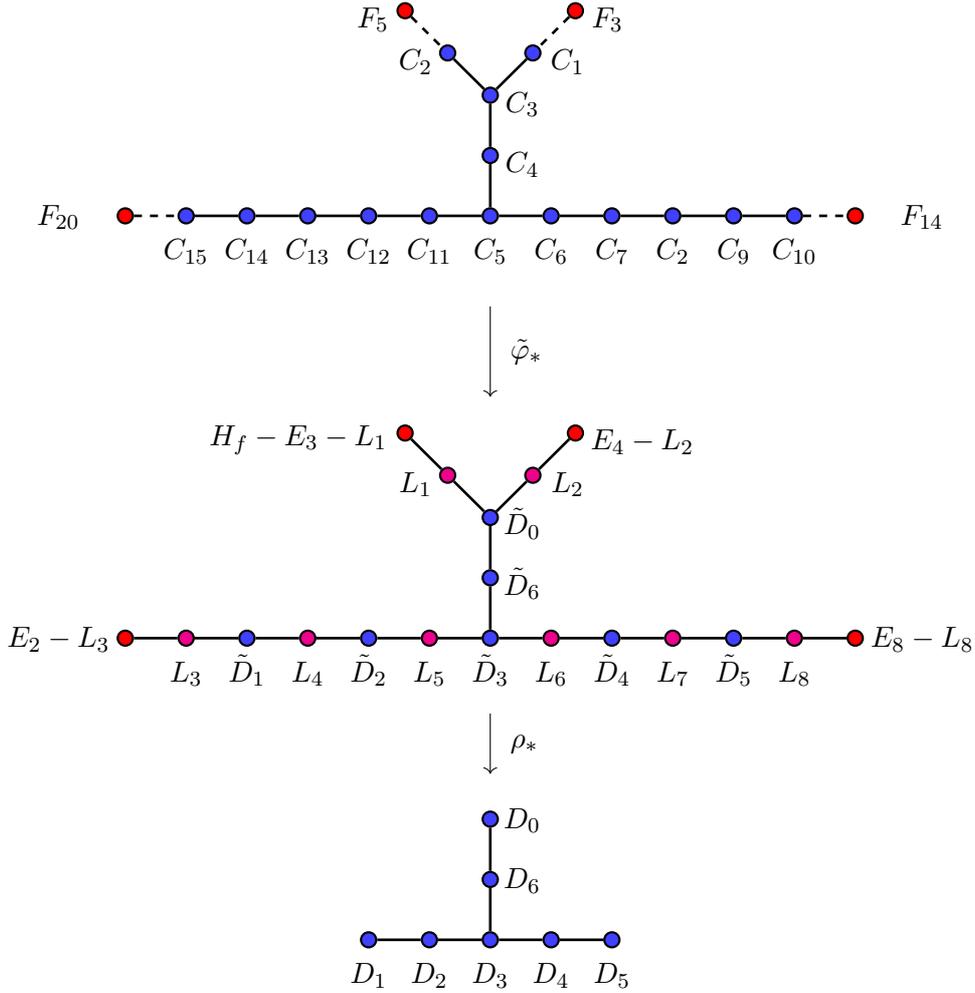
\begin{figure}[h] 
\centering
\begin{tikzpicture}[scale=.8, 
				elt/.style={circle,draw=black!100, fill=blue!75, thick, inner sep=0pt,minimum size=2mm},
				blk/.style={circle,draw=black!100, fill=black!75, thick, inner sep=0pt,minimum size=2mm}, 
				mag/.style={circle, draw=black!100, fill=magenta!100, thick, inner sep=0pt,minimum size=2mm}, 
				red/.style={circle, draw=black!100, fill=red!100, thick, inner sep=0pt,minimum size=2mm} ]
				
	\begin{scope}[xshift=0cm]	
		\path 	(-4,0) 	node 	(d1) [elt, label={[xshift=-0pt, yshift = -25 pt] $C_{14}$} ] {}
		        		(-2,0) 	node 	(d2) [elt, label={[xshift=-0pt, yshift = -25 pt] $C_{12}$} ] {}
		      		( 0,0) 	node  	(d3) [elt, label={[xshift=0pt, yshift = -25 pt] $C_5$} ] {}
		       		( 2,0) 	node  	(d4) [elt, label={[xshift=0pt, yshift = -25 pt] $C_7$} ] {}
		        		( 4,0) 	node 	(d5) [elt, label={[xshift=0pt, yshift = -25 pt] $C_9$} ] {}
		        		( 0,1) 	node 	(d6) [elt, label={[xshift=12pt, yshift = -15 pt] $C_4$} ] {}
		        		( 0,2) 	node 	(d0) [elt, label={[xshift=12pt, yshift = -15 pt] $C_3$} ] {};
		\path 	(.7,2.7)     node         (m1) [elt, label={[xshift=13pt, yshift = -15 pt] $C_1$} ] {}
				(-.7, 2.7)	node         (m2) [elt, label={[xshift=-12pt, yshift = -15 pt] $C_2$} ] {};
		\path 	(1,0) 		node		(m3)  [elt, label={[xshift=0pt, yshift = -25 pt] $C_6$} ] {}
				(3,0) 		node		(m4)  [elt, label={[xshift=0pt, yshift = -25 pt] $C_2$} ] {}
				(5,0) 		node		(m8)  [elt, label={[xshift=0pt, yshift = -25 pt] $C_{10}$} ] {};
		\path 	(-1,0) 	node		(m5)  [elt, label={[xshift=0pt, yshift = -25 pt] $C_{11}$} ] {}
				(-3,0) 	node		(m6)  [elt, label={[xshift=0pt, yshift = -25 pt] $C_{13}$} ] {}
				(-5,0) 	node		(m7)  [elt, label={[xshift=0pt, yshift = -25 pt] $C_{15}$} ] {};
		\path 	(1.4,3.4) 	node		(r1)  [red, label={[xshift=12pt, yshift = -15 pt] $F_3$} ] {}
				(-1.4,3.4) 	node		(r2)  [red, label={[xshift=-12pt, yshift = -15 pt] $F_5$} ] {};
		\path 	(6,0) 		node		(r3)  [red, label={[xshift=25pt, yshift = -12 pt] $F_{14}$} ] {}
				(-6,0) 	node		(r4)  [red, label={[xshift=-25pt, yshift = -12 pt] $F_{20}$} ] {};

		\draw [black,line width=1pt ] (m7) -- (d1) -- (m6) -- (d2) -- (m5) -- (d3);
		\draw [black,line width=1pt ] (d3) -- (m3)  -- (d4) -- (m4) -- (d5) -- (m8) ;
		\draw [dashed, line width=1pt ] (r4) -- (m7);
		\draw [dashed, line width=1pt ] (m8) -- (r3);
		\draw [black, dashed, line width=1pt ] (r1) -- (m1) ;		
		\draw [black, dashed, line width=1pt ] (r2) -- (m2) ;		
		\draw [black,line width=1pt ] (d3) -- (d6) -- (d0) ;		
		\draw [black,line width=1pt ] (d0) -- (m1) ;		
		\draw [black,line width=1pt ] (d0) -- (m2) ;		

	\end{scope}

	\draw [->] (0,-1.5)--(0,-3) node[pos=0.5, right] {~$\tilde{\varphi}_*$};

	\begin{scope}[yshift=-7cm]	
		\path 	(-4,0) 	node 	(d1) [elt, label={[xshift=-0pt, yshift = -25 pt] $\tilde{D}_1$} ] {}
		        		(-2,0) 	node 	(d2) [elt, label={[xshift=-0pt, yshift = -25 pt] $\tilde{D}_2$} ] {}
		      		( 0,0) 	node  	(d3) [elt, label={[xshift=0pt, yshift = -25 pt] $\tilde{D}_3$} ] {}
		       		( 2,0) 	node  	(d4) [elt, label={[xshift=0pt, yshift = -25 pt] $\tilde{D}_4$} ] {}
		        		( 4,0) 	node 	(d5) [elt, label={[xshift=0pt, yshift = -25 pt] $\tilde{D}_5$} ] {}
		        		( 0,1) 	node 	(d6) [elt, label={[xshift=12pt, yshift = -15 pt] $\tilde{D}_6$} ] {}
		        		( 0,2) 	node 	(d0) [elt, label={[xshift=12pt, yshift = -15 pt] $\tilde{D}_0$} ] {};
		\path 	(.7,2.7)     node         (m1) [mag, label={[xshift=13pt, yshift = -15 pt] $L_2$} ] {}
				(-.7, 2.7)	node         (m2) [mag, label={[xshift=-12pt, yshift = -15 pt] $L_1$} ] {};
		\path 	(1,0) 		node		(m3)  [mag, label={[xshift=0pt, yshift = -25 pt] $L_6$} ] {}
				(3,0) 		node		(m4)  [mag, label={[xshift=0pt, yshift = -25 pt] $L_7$} ] {}
				(5,0) 		node		(m8)  [mag, label={[xshift=0pt, yshift = -25 pt] $L_8$} ] {};
		\path 	(-1,0) 	node		(m5)  [mag, label={[xshift=0pt, yshift = -25 pt] $L_5$} ] {}
				(-3,0) 	node		(m6)  [mag, label={[xshift=0pt, yshift = -25 pt] $L_4$} ] {}
				(-5,0) 	node		(m7)  [mag, label={[xshift=0pt, yshift = -25 pt] $L_3$} ] {};
		\path 	(1.4,3.4) 	node		(r1)  [red, label={[xshift=25pt, yshift = -15 pt] $E_4-L_2$} ] {}
				(-1.4,3.4) 	node		(r2)  [red, label={[xshift=-40pt, yshift = -15 pt] $H_f-E_3-L_1$} ] {};
		\path 	(6,0) 		node		(r3)  [red, label={[xshift=25pt, yshift = -12 pt] $E_8-L_8$} ] {}
				(-6,0) 	node		(r4)  [red, label={[xshift=-25pt, yshift = -12 pt] $E_2-L_3$} ] {};

		\draw [black,line width=1pt ] (r4) -- (m7) -- (d1) -- (m6) -- (d2) -- (m5) -- (d3);
		\draw [black,line width=1pt ] (d3) -- (m3)  -- (d4) -- (m4) -- (d5) -- (m8) -- (r3);
		\draw [black,line width=1pt ] (r1) -- (m1) ;		
		\draw [black,line width=1pt ] (r2) -- (m2) ;		
		\draw [black,line width=1pt ] (d3) -- (d6) -- (d0) ;		
		\draw [black,line width=1pt ] (d0) -- (m1) ;		
		\draw [black,line width=1pt ] (d0) -- (m2) ;		
	\end{scope}
	
	\draw [->] (0,-8.25)--(0,-9.25) node[pos=0.5, right] {~$\rho_*$};

	\begin{scope}[yshift=-12cm]	
		\path 	(-2,0) 	node 	(d1) [elt, label={[xshift=-0pt, yshift = -25 pt] $D_{1}$} ] {}
		        (-1,0) node 	(d2) [elt, label={[xshift=-0pt, yshift = -25 pt] $D_{2}$} ] {}
		        ( 0,0) 	node  	(d3) [elt, label={[xshift=0pt, yshift = -25 pt] $D_{3}$} ] {}
		        ( 1,0) 	node  	(d4) [elt, label={[xshift=0pt, yshift = -25 pt] $D_{4}$} ] {}
		        ( 2,0) node 	(d5) [elt, label={[xshift=0pt, yshift = -25 pt] $D_{5}$} ] {}
		        ( 0,1) node 	(d6) [elt, label={[xshift=12pt, yshift = -12 pt] $D_{6}$} ] {}
		        ( 0,2) node 	(d0) [elt, label={[xshift=12pt, yshift = -12 pt] $D_{0}$} ] {};
		\draw [black,line width=1pt ] (d1) -- (d2) -- (d3) -- (d4)  -- (d5);
		\draw [black,line width=1pt ] (d3) -- (d6) -- (d0) ;		
	\end{scope}
	\end{tikzpicture}
	\caption{Behaviour of inaccessible divisors under the map $\varphi : \X^{\Tak}_m \rightarrow \X^{\Ok}$}
	\label{dynkinsforphi}
\end{figure}

\subsection{Symmetries}

In \cite{GLS} the following B\"acklund transformation for the fourth Painlev\'e equation is given. Let $\lambda=\lambda(t)$ be the solution of the fourth Painlev\'e equation with paramters $\alpha$ and  $\beta$, then a new solution $\lambda_1=\lambda_1(t)$ with parameters $\alpha_1$ and $\beta_1$ is given by
\begin{equation}
\lambda_1=\frac{\lambda'-\mu_1 \lambda^2-2\mu_1 t \lambda -\mu_2 b}{2\mu_1 \lambda}
\end{equation}
with 
\begin{equation}
\alpha_1=\frac{1}{4}(2\mu_1-2\alpha+3\mu_1\mu_2 b),\;\;\beta_1=-\frac{1}{2}(1+\mu_1\alpha+\mu_2 b/2)^2,
\end{equation}
where $\mu_1^2=\mu_2^2=1$ and $b^2=-2\beta$.
Using the change of variables we can easily verify the corresponding B\"acklund transformation for the Takasaki equation \eqref{takasakiscalar}. Let $q=q(t)$ be the solution with parameters $\alpha$ and $\beta$ then the new solution $q_1=q_1(t)$ with parameters $\alpha_1$ and $\beta_1$ as above is given by
\begin{equation}
q_1=\frac{\mu_3}{q}\sqrt{\frac{8q q'-\mu_1 q^4-8t \mu_1 q^2-16 \mu_2 b}{2\mu_1}}
\end{equation}
with additionally $\mu_3^2=1$.
Note that according to \cite[Th. 25.3]{GLS} we do not have a composition of B\"acklund transformations that would give $\beta=0$, so this formula does not provide a symmetry in the $\beta=0$ case. 

Alternatively we may consider the generators of the whole extended affine Weyl group $\widetilde{W}(A_2^{(1)})$ of symmetries of $\pain{IV}$ using their form as birational transformations of the variables $(f,g)$ from system \eqref{hamOkP4}. 
Lifting these under the transformation \eqref{TaktoOktransformation}, we find algebraic symmetries of the Takasaki system \eqref{takasakisystem}, in terms of both the original $(q,p)$ variables and $(x_1,y_1)$ from the atlas \eqref{takasakiatlas}. 
Note that for conciseness we consider only the symmetries of \eqref{hamOkP4} which fix the independent variable $t$, so elements of $\operatorname{Aut}(A_2^{(1)})$ of order two are excluded, though these can be lifted in a similar way.  
The symmetries of the Takasaki system take the form of a vector of algebraic relations 
$
\mathsf{F} (q,p,\tilde{q},\tilde{p}) = 0$  (resp. $\mathsf{G}(x_1,y_1,\tilde{x}_1,\tilde{y}_1) = 0$) ,
such that if $(q(t),p(t))$ (resp. $(x_1(t), y_1(t))$) solve the Takasaki system with parameters $a_i$, then $(\tilde{q}(t),\tilde{p}(t))$ (resp. $(\tilde{x}_1(t), \tilde{y}_1(t))$) solve the system with parameters $\tilde{a}_i$. 
We give the transformations of variables and parameters in \autoref{tab:symmetriestakasaki} corresponding to the generators $s_0, s_1, s_2$, and $\rho = \pi_2 \pi_1$ of the symmetry group written in \cite[Sec. 8.4.20]{KNY}, where $a_0 = 1 - a_1 - a_2$ is introduced for convenience.

\def\arraystretch{2}
\begin{table}[h]
  \centering 
$	\begin{array}{ c  | c  | c  | c | c }
 		 ~ & \mathsf{F} (q,p,\tilde{q},\tilde{p})=0 & \mathsf{G}(x_1,y_1,\tilde{x}_1,\tilde{y}_1)=0 & \tilde{a}_1 & \tilde{a}_2 \\
		 \hline
		 \multirow{3}{2em}{\centering $s_0$} 
		 & \tilde{q}^2 = q^2 \left( 1 - \frac{64 a_0}{\mathcal{Q}} \right) & \tilde{x}_1 = x_1 - \frac{8 a_0}{y_1^2 - 4 x_1 +8 t} & \multirow{3}{4em}{\centering $1-a_2$} & \multirow{3}{4em}{\centering $1 - a_1$}\\ 					
		&\tilde{q} \tilde{p} = q p + 4a_0 \left( 1 - \frac{4 q^2(q^2 + 4 t )}{\mathcal{Q}} + \frac{128 a_0 q^4}{\mathcal{Q}^2} \right) &  \tilde{y}_1^2 = y_1^2 - \frac{32 a_0}{y_1^2 - 4 x_1 +8 t}  &  & \\
		& \mathcal{Q}= q^4 + 8 t q^2 - 8 q p - 32 a_1& & &   \\
	\hline					
		 \multirow{2}{2em}{\centering $s_1$} & \tilde{q}^2=q^2 & \tilde{x}_1  = \frac{x_1 y_1^2 - 8 a_1}{y_1^2},  &  \multirow{2}{4em}{\centering $- a_1$} &  \multirow{2}{4em}{\centering $a_1 + a_2$} \\
		 					& \tilde{q}\tilde{p}=q p &  \tilde{y}_1^2 = y_1^2 & &  \\
	\hline
		 \multirow{2}{2em}{\centering $s_2$} & \tilde{q}^2 = q^2 \left( 1 + \frac{64 a_2}{q^4 + 8 t q^2 + 8 q p + 32 a_1} \right)& \tilde{x}_1 = x_1 &  \multirow{2}{4em}{\centering $a_1 + a_2$}&  \multirow{2}{4em}{\centering $- a_2$} \\
		 					& \frac{\tilde{q}^4 + 8 \tilde{q} \tilde{p} + 32 (a_1+a_2)}{\tilde{q}^2} =  \frac{q^4 + 8 q p + 32 a_1}{q^2} & \tilde{y}_1^2 = y_1^2 + \frac{8 a_2}{x_1} &  & \\
	\hline
		 \multirow{2}{2em}{\centering $\rho$} & \tilde{q}^2 = - \frac{q^4 + 8 t q^2 + 8 q p + 32a_1}{2 q^2} & \tilde{y}_1^2 = - 4 x_1 & \multirow{2}{4em}{\centering $1- a_1 -a_2$} & \multirow{2}{4em}{\centering $a_1$}  \\
		 					& \tilde{p} =  \frac{ \tilde{q}(q^4 - 8 q p - 32 a_1)}{8 q^2}-  \frac{32a_0 - \tilde{q}^4}{8\tilde{q}}  & \tilde{x}_1 = \frac{y_1^2 - 4 x_1 + 8 t }{4} &  & \\
		\hline 
  	\end{array}
$	
  \caption{Algebraic B\"acklund transformation symmetries of the Takasaki system}\label{tab:symmetriestakasaki}
\end{table}
\def\arraystretch{1}


Among the symmetries of the Takasaki system we may consider one corresponding to the discrete Painlev\'e equation appearing in \cite[Sec. 8.1.18]{KNY}, which we have rescaled such that it provides a B\"acklund transformation for \eqref{hamOkP4}.
The discrete system in question is 
\begin{equation}
f_{n+1} = 2 g_n - f_n - \frac{a_2}{g_n} - 2 t,  \qquad g_{n+1} = \frac{1}{2} f_{n+1} - g_n + \frac{\bar{a}_1}{f_{n+1}} + t, 
\end{equation}
where the parameters are now $n$-dependent, denoted $a_i = a_i(n)$, $\bar{a}_i = a_i(n+1)$, and evolve with the discrete dynamics according to 
\begin{equation} \label{dPparams}
\bar{a}_{1} = a_{1} - 1, \qquad \bar{a}_{2} = a_{2} + 1.
\end{equation}
Lifting this in the same way as above we obtain
\begin{equation}
\begin{gathered}
\bar{q}^2 = \frac{64 a_2 q^2}{q^4 + 8 t q^2 + 8 q p + 32 a_1} - \frac{q^4 + 8 t q^2 - 8 q p - 32 a_1}{2 q^2}   ,\\
 \frac{\bar{q}^4 + 32 \bar{a}_1}{ 8 \bar{q}^2} - \frac{\bar{p}}{\bar{q}} =  \frac{q^4 + 32 a_1 }{8 q^2} + \frac{p}{q},
\end{gathered}
\end{equation}
where $\bar{q} = q_{n+1}$, $q = q_{n}$, with parameter evolution given by \eqref{dPparams}.

\section{Regularisation of algebro-Painlev\'e systems on rational surfaces} \label{section4}

In this section we show that the mechanism by which the Takasaki system is associated to the surface $\X^{\Tak}$, namely quadratic regularisability in the sense of \autoref{defquadreg}, can be generalised to give a similar mechanism to associate equations with globally finite branching of solutions about movable singularities to rational surfaces. 
We introduce a notion of regularisability by algebraic transformations and illustrate this phenomenon in a number of examples with the so-called algebro-Painlev\'e property.
\begin{definition}
Consider an $n$-th ordinary differential equation 
\begin{equation}\label{nthorderODE}
\frac{d^n y}{dt^n} = F( y, \frac{dy}{dt}, \dots, \frac{d^{n-1} y}{dt^{n-1}} ; t ),
\end{equation}
where $F$ is rational in $y$ and its derivatives, and locally analytic in $t$. Let the fixed singularities of the equation be the discrete set $\mathcal{F} \subset \C$, and let $B =\C \backslash \mathcal{F}$. The equation \eqref{nthorderODE} is said to have the algebro-Painlev\'e property if all solutions are algebroid functions over $B$, by which we mean they are algebraic over the field of meromorphic functions on the universal cover of $B$. 
\end{definition}
The notion of the quadratic regularisability in \autoref{defquadreg} can be naturally extended as follows.
\begin{definition} \label{monomialregdef}
Suppose a first-order system of ODEs is of the form
\begin{equation*}
u' = R_1( t,u, v^n ), \quad v' = \frac{1}{v^{n-1}} R_2 (t,u, v^n),
\end{equation*}
with $R_1(t,x,y),$ $R_2(t,x,y)$ rational functions of their arguments, and $R_1(t,x,y),$ $R_2(t,x,y)$ being regular  in $y$  and nonzero at $y=0$. 
Then we say that such a system is $n$\emph{-th order monomial regularisable} in the variable $v$. Introducing $h = v^n$, the resulting system is regular in $h$ at $h=0$:
\begin{equation*}
u' = R_1(t, u, h ), \quad h'=  nR_2 (t,u, h).
\end{equation*}
\end{definition}

\subsection{Example 1} \label{section4.1}

In \cite{HK}, Halburd and Kecker considered a class of second-order equations, and isolated equations whose solutions are globally quadratic over the field of meromorphic functions. 
Under certain assumptions on the coefficients it is shown that if 
\begin{equation} \label{HKclass}
y''=\frac{3}{4}y^5+\sum_{k=0}^4 a_k(t)y^k
\end{equation}
and $y^2(t)+s_1(t)(t)+s_2(t)=0$, then $s_1$ is proportional to $a_4$ (which can be set to zero without loss of generality) and $s_2$ reduces either to a solution of a Riccati equation or the equation 
\begin{equation} \label{specialP4}
w''=\frac{w'^2}{2w}+\frac{3}{2}w^3+4(a t+b)w^2+2((a t+b)^2-c)w,
\end{equation}
which in case $a\neq 0$ is equivalent to a special case of the fourth Painlev\'e equation and in case $a=0$ can be solved in terms of elliptic functions.
We call the equation isolated in \cite{HK} from the class \eqref{HKclass} the Halburd-Kecker equation, which is given explicitly by
\begin{equation}\label{HK eq}
y''=\frac{3}{4}y^5-(2a t+2b) y^3+\left( (a t+b)^2 - 2c \right)y.
\end{equation}
By taking $s_2=-y^2$ one recovers the variant \eqref{specialP4} of the fourth Painlev\'e equation above.
Note that if one follows the approach of \cite{HK} but assuming $s_1\neq 0$, then $s_2$ must satisfy an equation of second order and second degree. 
The question of how to write such an equation in the form of an equivalent system of two first order differential equations with rational right-hand sides is open, so this case falls beyond the scope of the current paper. 

It can be easily seen that the Halburd-Kecker equation is equivalent to the $\beta=0$ case of the scalar Takasaki equation:

\begin{proposition}
The equation \eqref{HK eq} for $y(\tilde{t})$ with parameters $a$, $b$, $c$ is equivalent to equation \eqref{takasakiscalar} with $\beta=0$ for $q(t)$ via the change of variables
\begin{equation}
y(\tilde{t}) = A q(t), \quad B t = a \tilde{t} + b, \qquad \text{where} \qquad A^2 = - \frac{B}{4}, \quad B^2 = a, \quad \alpha = \frac{2 c}{a}.
\end{equation}
\end{proposition}

By putting $q= y$, $p = q'$, we may rewrite the Halburd-Kecker equation as an equivalent first-order system 
\begin{equation} \label{halburdkeckersystem}
\begin{aligned}
\frac{dq}{dt} &=  p, \\
 \frac{dp}{dt} &= \frac{3}{4} q^5 - 2(a t + b) q^3 + \left((at+b)^2 - 2 c\right) q .
\end{aligned}
\end{equation}
Considering this on the trivial bundle over $\C_t$ with fibre $\p^1 \times \p^1$ in the same way as for the Takasaki system in \autoref{appendixB}, we find initially only one point $(Q,P)=(1/q, 1/p) =(0,0)$ that requires blowing up, so the five blowups over $(q,P)=(0,0)$ are not required here. 
To resolve the indeterminacy of the vector field require essentially the same sequence of blowups over this point as in the Takasaki case, i.e. three blowups of points on the proper transform of $P=0$, followed by two cascades of six points each.
The configuration of points is the same as in \autoref{appendixB} but without the blowups of $q_1,\dots, q_5$, so there are 15 blowups performed in total but we maintain the enumeration of the remaining points from the Takasaki case, so there are initially three repeated blowups before the cascade splits:
\begin{equation}
q_6 : (Q,P) = (0,0) ~ \leftarrow ~ q_7 \leftarrow ~ q_8 ~ \leftarrow ~
\left\{
\begin{aligned}
q_9 &\leftarrow  q_{10} \leftarrow  \cdots  \leftarrow q_{14},  \\
q_{15} &\leftarrow  q_{10} \leftarrow  \cdots  \leftarrow q_{20}.  
\end{aligned}
\right.
\end{equation}
Through these blowups we obtain a surface $\X^{\operatorname{HK}}$, with the precise locations of points in coordinates of the points being the only difference from the Takasaki case. 
We present only the charts for the pair of final exceptional divisors, which together with the original variables $(q,p)$ provide an atlas for a manifold on which the Halburd-Kecker system \eqref{halburdkeckersystem} is regularisable.
\begin{proposition}
The Halburd-Kecker system \eqref{halburdkeckersystem} is everywhere either regular or quadratic regularisable on the bundle
\begin{equation} \label{atlasHK}
E^{\operatorname{HK}} = \C^3_{q,p,t} \cup \C^3_{u_{14},v_{14},t}\cup \C^3_{u_{20},v_{20},t},
\end{equation}
with gluing defined according to
\begin{equation} \label{HKatlas}
\begin{gathered}
\frac{1}{q} = v_{14}, \quad \frac{1}{p} = v_{14}^3\left( 2 + v_{14}^2 \left(4(at + b) + v_{14}^2 \left( 8(b^2 + c) + 4 a (2 a t^2+ 4 b t - 1) + u_{14} v_{14}^2 \right)  \right) \right),\\
\frac{1}{q} = v_{20}, \quad \frac{1}{p} = v_{20}^3\left(-2 + v_{20}^2 \left( - 4(at + b) + v_{20}^2 \left( - 8(b^2 + c)  - 4 a (2 a t^2+ 4 b t - 1)+ u_{20} v_{20}^2 \right)  \right) \right).
\end{gathered}
\end{equation}
\end{proposition}
In particular the system rewritten in the variables $(u_{14}, v_{14})$, respectively $(u_{20},v_{20})$, satisfies \autoref{monomialregdef} with $n=2$.

\subsection{Example 2} \label{section4.2}
Let us consider the following equation 
\begin{equation}\label{cubic1}
3q^2q''+\frac{3}{2}q q'^2-2Bq^6+B t q^3+\frac{A}{2q^3}=0,
\end{equation}
where $A$ and $B$ are arbitrary constants which we assume to be nonzero.
The local expanson at $t_0$ is $q\sim (t-t_0)^{-2/3}$.  By a cubic change of variables $w=q^3$ it is reduced to the Painleve XXXIV equation
\begin{equation*}
w''=\frac{w'^2}{2w}+B w(2w-t)-\frac{A}{2w},
\end{equation*}
which possesses the Painlev\'e property, so the algebraic transformation ensures that equation \eqref{cubic1} has the algebro-Painlev\'e property.
By introducing $q'=p$ we can easily re-write equation \eqref{cubic1} in the form of a system of two first order differential equations for $q$ and $p$ with rational right-hand sides.
Extending this to $\p^1 \times \p^1$ in the usual way, we require 24 blowups in total. The points are grouped into a cascade of eight points over $(q,p) = (0, \infty)$ which splits in two after the second blowup according to 
\begin{equation}
p_1 : (q,p) = (0,\infty)  \leftarrow  p_2 \ \leftarrow 
\left\{
\begin{aligned}
p_3 &\leftarrow  p_4 \leftarrow  p_5,  \\
p_6 &\leftarrow p_7 \leftarrow  p_8,  
\end{aligned}
\right.
\end{equation}
and one long cascade $p_9 \leftarrow \cdots \leftarrow p_{24}$ of 16 points over $p_9 : (q,p)=(\infty,\infty)$.
After this we see that the resulting 3 systems in the final charts are all cubic regularisable, satisfying \autoref{monomialregdef} with $n=3$. 

\begin{proposition}
The system equivalent to equation \eqref{cubic1} via $q'=p$ is everywhere either regular or cubic regularisable on the bundle
\begin{equation} \label{atlasexample2}
\left( \C^3_{q,p,t}\backslash \left\{q=0 \right\} \right) \cup \C^3_{u_{5},v_{5},t} \cup \C^3_{u_{8},v_{8},t} \cup \C^3_{u_{24},v_{24},t} ,
\end{equation}
with gluing defined according to
\begin{equation} 
\begin{gathered}
q = v_{5}, \quad \frac{1}{p} = v_{5}^2 \left( \frac{3}{\sqrt{A}} + u_5 v_5^{3} \right), \\
q = v_{8}, \quad \frac{1}{p} = v_{8}^2 \left(- \frac{3}{\sqrt{A}} + u_8 v_8^{3} \right), \\
\frac{1}{q} = \frac{v_{24}^2}{3^{11}}  \left(3366 B - 432 B^4 v_{24}^6+ 64 B^5 v_{24}^9 + 177147 u_{24} v_{24}^{12} \right), \\
\frac{1}{p} = \frac{v_{24}^5}{3^{22}}  \left(3366 B - 432 B^4 v_{24}^6+ 64 B^5 v_{24}^9 + 177147 u_{24} v_{24}^{12} \right)^2.
\end{gathered}
\end{equation}
\end{proposition}

\subsection{Example 3} \label{section4.3}
Though the examples so far have been related by $t$-independent algebraic transformations to second-order equations with the Painlev\'e property, we now present an example related by a $t$-dependent transformation to the fourth Painlev\'e equation.
The equation is
\begin{equation} \label{example3scalar}
\begin{aligned}
(q^3 - t^3) q'' &=  \left( \frac{2 t^3}{q} - \frac{1}{2} q^2\right) (q')^2 - 3t^2 q' + \frac{1}{2}q^{10} - \left(2 t^3 - \frac{4}{3}t \right) q^7 \\
&\quad + \left( 3 t^6 - 4 t^4 + \frac{2}{3}t^2 - \frac{2 \alpha}{3} \right) q^4 - t \left(2 t^8 - 4 t^6 + \frac{4}{3} t^4- \frac{4 \alpha}{3} t^2 - 2\right) q \\
&\quad \quad + \left(  \frac{1}{2} t^{12} - \frac{4}{3} t^{10} + \frac{2}{3} t^8 - \frac{2 \alpha}{3} t^6 - \frac{1}{2} t^4 + \frac{\beta}{3} \right) \frac{1}{q^2},
\end{aligned}
\end{equation}
whose local expanson at $t_0$ is $q\sim (t-t_0)^{-1/3}$. 
This was obtained by requiring that a cubic change of variables 
\begin{equation}
\lambda = q^3 - t^3
\end{equation}
transforms \eqref{example3scalar} to the fourth Painleve equation in the form \eqref{P4scalar}. 
We rewrite this as the first-order system
\begin{equation} \label{example3system}
\begin{aligned}
q' &= - \frac{q^4}{3} + \frac{2t (t^2-1) q}{3}  + \frac{3 t^2 + 2 t^4 - t^6 - 2 a_1}{3 q^2} + p \left( \frac{4 q}{3}  - \frac{4 t^3}{3 q^2} \right), \\
p' &= -2 p^2 + 2 p(q^3 - t^3 + t) + a_2 
\end{aligned}
\end{equation}
again making use of the parameters $a_0$, $a_1$, $a_2$ defined by $\alpha = 1- a_1-2a_2$, $\beta = - 2 a_1^2$, and $a_0 = 1 - a_1 - a_2$.
This case requires 24 blowups of $\p^1 \times \p^1$ to achieve regularisability, grouped as one cascade of six infinitely near points $p_1 \leftarrow \cdots \leftarrow p_6$, three double points, $p_7 \leftarrow p_8$, $p_9 \leftarrow p_{10}$, $p_{11} \leftarrow p_{12}$,  and finally a cascade of twelve points $p_{13} \leftarrow \cdots \leftarrow p_{24}$. 
\begin{proposition}
The system \eqref{example3system} is everywhere either regular or cubic regularisable on the bundle
\begin{equation} \label{atlasexample3}
 \C^3_{q,p,t}  \cup \C^3_{u_{6},v_{6},t} \cup \C^3_{u_{8},v_{8},t} \cup \C^3_{u_{10},v_{10},t} \cup \C^3_{u_{12},v_{12},t} \cup \C^3_{u_{24},v_{24},t} ,
\end{equation}
with gluing defined as follows, in which $\zeta$ is a primitive cube root of unity:
\begin{equation} 
\begin{gathered}
\frac{1}{q} = v_{6}, \quad p = v_{6}^3 \left( - a_2 + u_6 v_6^{3} \right), \\
q =  t + v_{8}, \quad \frac{1}{p} = v_{8} \left( \frac{3 t^2}{a_1} + u_8 v_8 \right), \\
q = \zeta t +  v_{10}, \quad \frac{1}{p} = v_{10} \left( \frac{3 \zeta^{-1} t^2}{a_1} + u_{10} v_{10} \right), \\
q = \zeta^{-1} t +  v_{12}, \quad \frac{1}{p} = v_{12} \left( \frac{3 \zeta t^2}{a_1} + u_{12} v_{12} \right), \\
\frac{1}{q} = v_{24}, \quad \frac{1}{p} = v_{24}^3 \left( 2+ 2 t (t^2 - 2) v_{24}^3 + 2 \left( t^2(t^2-2)^2 + 2 a_0 \right) v_{24}^6 + u_{24} v_{24}^9  \right).
\end{gathered}
\end{equation}
In this case, the system is genuinely regular on the parts of the final exceptional divisors in each of the three charts $(u_8,v_8)$, $(u_{10},v_{10})$, and $(u_{12},v_{12})$, while in the other charts it is cubic regularisable in $v_{6}$, respectively $v_{24}$.
\end{proposition}



%
%

\subsection{Example 4} \label{section4.4}

We now consider an example which is not transformable to any of the differential Painlev\'e equations but still possesses the algebro-Painlev\'e property. 
The equation is given by 
\begin{equation}
q'' = \frac{2 q^7}{3} - 2 t q^4 + 2 \left(t^2 + \frac{A}{3} \right) q - \frac{ 2(t^3 + A t +B)}{3 q^2} - \frac{2 (q')^2}{q},
\end{equation}
and is equivalent to the first-order system
\begin{equation} \label{example4system}
\begin{aligned}
q' &= - \frac{q^4}{3} + \frac{2t q }{3} + \frac{1 - A - t^2}{3 q^2} + \frac{2 p }{3 q^2}, \\
p' &= 2 p q^3 - 2 t p - B,
\end{aligned}
\end{equation}
which was obtained by letting $f = q^3 - t$,  $g=p$ in 
\begin{equation}\label{H2autosystem}
f' = 2 g - f^2 - A, \qquad g' = 2 f g - B,
\end{equation}
where $A$ and $B$ are arbitrary constants.
The system \eqref{H2autosystem} is an autonomous limit of the Okamoto Hamiltonian form of the second Painlev\'e equation, and in particular is Liouville integrable and may be solved in terms of elliptic functions.

For system \eqref{example4system} we again require 24 blowups, consisting of one cascade of six points $p_1 : (q,p)= (\infty, 0) \leftarrow p_2 \leftarrow \cdots \leftarrow p_6$, as well as another cascade of 18 points $p_7 : (q,p)= (\infty, \infty) \leftarrow p_{8} \leftarrow \cdots \leftarrow p_{24}$.

\begin{proposition}
The system \eqref{example4system} is everywhere either regular or cubic regularisable on the bundle
\begin{equation} \label{atlasexample4}
 \C^3_{q,p,t}  \cup \C^3_{u_{6},v_{6},t} \cup \C^3_{u_{24},v_{24},t} ,
\end{equation}
with gluing defined by
\begin{equation} 
\begin{gathered}
\frac{1}{q} = v_{6}, \quad p = v_{6}^3 \left(B + u_6 v_6^{3} \right), \\
\frac{1}{q} = v_{24}, \quad \frac{1}{p} = v_{24}^6 + 2 t v_{24}^9 + v_{24}^{12} \left( 3 t^2 - A + (4 t^3 - 4 A t + B) v_{24}^3 + u_{24} v_{24}^6 \right).
\end{gathered}
\end{equation}
\end{proposition}

 \subsection{Algebraic regularisability as an algebro-Painlev\'e test}

We now show that the requirement of algebraic regularisability in the sense of \autoref{monomialregdef} of a system of rational first-order differential equations after resolution of indeterminacies can serve as a kind of counterpart to the Painlev\'e test to isolate equations with algebroid solutions.
%
We first show that requiring quadratic regularisability after the same sequence of blowups recovers the Takasaki system, in the same spirit as the uniqueness results for global Hamiltonian structures of differential equations on Okamoto's spaces \cite{MMT99, ST97, M97, chiba, orbifold} which show in essence that the only globally regular Hamiltonian system that can exist on $E_{\mathrm{J}}$ is the extension of the Okamoto Hamiltonian form of $\pain{J}$.
\begin{theorem} \label{proprecoverTak}
If the system 
\begin{equation} \label{Takgen}
q'=p,\,p'=\frac{3q^5}{64}+\sum_{k=-3}^4 a_k(t)q^k
\end{equation}
is quadratic regularisable on a bundle of rational surfaces obtained by blowing up points in the same configuration as $\X^{\Tak}$, then it must coincide with the Takasaki system up to affine changes of independent variable.
\end{theorem}
\begin{proof}
We   start by resolving indeterminacies of system \eqref{Takgen} maintaining genericity of assumptions on the unknown coefficients $a_k(t)$. We only assume that $a_{-3}\neq 0$.  
We first find a sequence of five points of indeterminacy to blow up (with the same configuration as in the Takasaki case as indicated on the left-hand side of \autoref{fig:takasakisurface}), the locations of which are as follows (below $i^2 =-1$):
\begin{equation*}
\begin{gathered}
q_1 : (q,P) = (0,0) \quad \leftarrow \quad
\left\{
\begin{aligned}
q_2 : (u_1,v_1) &= (-i\sqrt{a_{-3}(t)}, 0) ~&\leftarrow &\quad q_3 : (u_2,v_2) = (-a_{-2}(t),0) \\
q_4: (u_1,v_1) &= (i\sqrt{a_{-3}(t)}, 0) ~&\leftarrow &\quad q_5 : (u_4, v_4) = (-a_{-2}(t),0)
\end{aligned}
\right.
\end{gathered}
\end{equation*}
After blowing up $q_3$ and $q_5$, there are no further points of indeterminacy over $q_1$.

We next find the following sequence of points over $(Q,P)=(0,0)$, which splits after the blowup of $q_8$ just as in the Takasaki case:
\begin{equation*}
\begin{gathered}
q_6 : (Q,P) = (0,0) ~ \leftarrow ~ q_7 : (U_6,V_6)= (0,0) ~ \leftarrow ~ q_8 : (u_7,v_7) = (0,0) ~ \leftarrow ~
\left\{
\begin{aligned}
&q_9  \\
&q_{15}
\end{aligned}
\right.
\end{gathered}
\end{equation*}
The cascade over the point $p_9$ is given by 
\begin{equation*}
\begin{aligned}
& & &p_9 : (u_8,v_8) = (8,0) \\
&  & &  \uparrow  & \\
& &  &p_{10} : (u_9, v_9) =\left(F_1(a_4(t)),0\right) \\
&  & &  \uparrow  & \\
&  & & p_{11} : (u_{10},v_{10})  =\left( F_{2}(a_3(t),a_4(t)), 0\right)  \\
&  & &  \uparrow  & \\
&  & & p_{12} : (u_{11},v_{11}) = \left( F_{3}(a_2(t),a_3(t),a_4(t),a_4'(t)), 0\right)  \\
&  & &  \uparrow  & \\
&  & & p_{13} : (u_{12},v_{12}) = \left( F_{4}(a_1(t), a_2(t), a_3(t), a_3'(t), a_4(t),a_4'(t)), 0\right)  \\
&  & &  \uparrow  & \\
&  & & p_{14} : (u_{13},v_{13}) = \left( F_{5}\left(a_0(t), a_1(t), a_2(t), a_2'(t), a_3(t), a_3'(t), a_4(t),a_4'(t), a_4''(t)\right), 0\right), 
\end{aligned}
\end{equation*}
where the functions $F_9,\dots, F_{13}$ are known polynomial functions of their arguments which we omit for conciseness.
The cascade over the point $p_{15}$ is similar, given by
\begin{equation*}
\begin{aligned}
& & &p_{15} : (u_8,v_8) = (8,0) \\
&  & &  \uparrow  & \\
& &  &p_{16} : (u_{15}, v_{15} =\left(G_1(a_4(t)),0\right) \\
&  & &  \uparrow  & \\
&  & & p_{17} : (u_{16},v_{16})  =\left( G_{2}(a_3(t),a_4(t)), 0\right)  \\
&  & &  \uparrow  & \\
&  & & p_{18} : (u_{17},v_{17}) = \left( G_{3}(a_2(t),a_3(t),a_4(t),a_4'(t)), 0\right)  \\
&  & &  \uparrow  & \\
&  & & p_{19} : (u_{18},v_{18}) = \left( G_{4}(a_1(t), a_2(t), a_3(t), a_3'(t), a_4(t),a_4'(t)), 0\right)  \\
&  & &  \uparrow  & \\
&  & & p_{20} : (u_{19},v_{19}) = \left( G_{5}\left(a_0(t), a_1(t), a_2(t), a_2'(t), a_3(t), a_3'(t), a_4(t),a_4'(t), a_4''(t)\right), 0\right), 
\end{aligned}
\end{equation*}
where the functions $G_1,\dots, G_5$ are known polynomials in their arguments.
After the blowups of $q_1, \dots, q_{20}$, there are no indeterminacies remaining, so we see that without additional assumptions on the coefficients we have the same point configuration as $\X^{\Tak}$.

Next we need to impose the condition of quadratic-regularisability on the systems in coordinate charts covering the exceptional divisors from the blowups of points $q_3$, $q_5$, $q_{14}$ and $q_{20}$. 
In general, whenever we impose any conditions on the coefficients $a_k(t)$ we need to check that the points of indeterminacy do not disappear (and also new points do not appear) and our calculations are still valid (the locations of points might change but the configuration stays the same).  
We do so  every time by running the calculations from the beginning.
 
The system in coordinates $u_{3}$ and $v_{3}$ after the blowup of the point $q_{3}$ is of the form
\begin{equation}
u_{3}'=\frac{P^{(1)}_{3}(t,u_{3},v_{3})}{v_{3}^2 Q^{(1)}_{3}(t,u_{3},v_{3})}, \qquad v_{3}'=\frac{P^{(2)}_{3}(t,u_{3},v_{3})}{v_{3} Q^{(2)}_{3}(t,u_{3},v_{3})}
\end{equation}
with $P^{(j)}_{3}(t,u_{3},0)$ and $ Q^{(j)}_{3}(t,u_{3},0)$ ($j=1,\,2$) nonzero and independent of $u_{3}$. 
To force the cancellation of a factor of $v_3$ in the rational function giving $u_3'$, we must set $P^{(1)}_{3}(t,u_{3},0)=0$, leading to 
\begin{equation}\label{cond1}
a_3(t)^{3/2} (2 \sqrt{a_{-3}(t)} a_{-1}(t) + i a_{-3}'(t))=0.
\end{equation}
After imposing this condition we confirm that the system above becomes 
\begin{equation}
u_{3}'=\frac{\tilde{P}^{(1)}_{3}(t,u_{3},v_{3})}{v_{3} \tilde{Q}^{(1)}_{3}(t,u_{3},v_{3})}, \qquad v_{3}'=\frac{\tilde{P}^{(2)}_{3}(t,u_{3},v_{3})}{v_{3} \tilde{Q}^{(2)}_{3}(t,u_{3},v_{3})},
\end{equation}
so one power of $v_3$ in the denominator in the first equation has cancelled as required. 

Now, to make $u_3'$ regular and nonzero in $v_3$ at $0$, we must set $\tilde{P}^{(1)}_{3}(t,u_{3},0)=0$. 
Collecting the coefficients of $u_3$ in this expression we  find that $a_{-2}(t)=0$ and $a_{0}(t)=0$ (recall that $a_{-3}(t)\neq 0$). 
Another condition that we need to impose is that $u_3'$ is even in $v_3$, which leads to $a_2(t)=a_4(t)=0$. 
Combining this with the previously determined coefficients $a_0(t)$, $a_{-2}(t)$, which are   equal to zero, and $a_{-3}'(t)$ from (\ref{cond1}), we find that the system in $u_3$ and $v_3$ coordinates is quadratic-regularisable, that is, we also have that expression $v_3 v_3'$  is  regular and nonzero at $v_3=0$ and even in $v_3$. 

We could also have begun with the system in coordinates $u_{5}$ and $v_{5}$ after the last blowup of the point $p_{5}$, which with generic coefficient functions $a_k(t)$  is of the form
\begin{equation}
u_{5}'=\frac{P^{(1)}_{5}(t,u_{5},v_{5})}{v_{5}^2 Q^{(1)}_{5}(t,u_{5},v_{5})}, \qquad v_{5}'=\frac{P^{(2)}_{5}(t,u_{5},v_{5})}{v_{5} Q^{(2)}_{5}(t,u_{5},v_{5})}
\end{equation}
with $P^{(j)}_{5}(t,u_{5},0)$ and $ Q^{(j)}_{5}(t,u_{5},0)$ ($j=1,\,2$) nonzero and independent of $u_{5}$. 
As above, quadratic-regularisability requires us to set the  expression $P^{(1)}_{5}(t,u_{5},0)=0$, which leads to 
\begin{equation}\label{cond2}
a_3(t)^{3/2} (2 \sqrt{a_{-3}(t)} a_{-1}(t) - i a_{-3}'(t))=0.
\end{equation}
After imposing this condition the system above would simplify to 
\begin{equation}
u_{5}'=\frac{\tilde{P}^{(1)}_{5}(t,u_{5},v_{5})}{v_{5} \tilde{Q}^{(1)}_{5}(t,u_{5},v_{5})}, \qquad v_{5}'=\frac{\tilde{P}^{(2)}_{5}(t,u_{5},v_{5})}{v_{5} \tilde{Q}^{(2)}_{5}(t,u_{5},v_{5})},
\end{equation}
so again one power of $v_5$ in the denominator in the first equation has disappeared. 
Now, to make $u_5'$ regular and nonzero in $v_5$ at $0$, we must set $\tilde{P}^{(1)}_{5}(t,u_{5},0)=0$. Collecting the coefficients of $u_5$ in this expression we  find that $a_{-2}(t)=0$, $a_{0}(t)=0$. 
Thus, we obtain the same conditions as above for system in coordinates $u_3$, $v_3$.

Next we impose condition that $u_5'$ is even in $v_5$. This leads again to $a_2(t)=a_4(t)=0$. Under this assumptions together with previously determined coefficients $a_0(t)$, $a_{-2}(t)$, which are all equal to zero, and $a_{-3}'(t)$ from (\ref{cond2}), we find that the system in $u_5$ and $v_5$ coordinates is quadratic-regularisable, that is, we also have that the expression $v_5 v_5'$  is  regular and nonzero at $v_5=0$ and even in $v_5$. 

Adding and subtracting  (\ref{cond1}) and (\ref{cond2}) yields $a_{-3}'(t)=0$ and $a_{-1}(t)=0$, so for quadratic regularisability in the charts $(u_3,v_3)$ and $(u_5,v_5)$ we must have 
\begin{equation} \label{condssofar}
a_{-2}(t)=a_{-1}(t)=a_0(t)=a_2(t)=a_4(t)=a_{-3}'(t)=0.
\end{equation}
Without any specialisation of coefficient functions, the system in the coordinates $u_{14}$ and $v_{14}$ after the last blowup of the point $p_{14}$ is of the form
\begin{equation} 
u_{14}'=\frac{P^{(1)}_{14}(t,u_{14},v_{14})}{v_{14}^2 Q^{(1)}_{14}(t,u_{14},v_{14})},\qquad  v_{14}'=\frac{P^{(2)}_{14}(t,u_{14},v_{14})}{v_{14} Q^{(2)}_{14}(t,u_{14},v_{14})}
\end{equation}
with $P^{(j)}_{14}(t,u_{14},0)$ and $ Q^{(j)}_{14}(t,u_{14},0)$ ($j=1,\,2$) nonzero and independent on $u_{14}$. Moreover, if we set  expression $P^{(1)}_{14}(t,u_{14},0)=0$, which gives some (cumbersome) relation between coefficients $a_k(t)$ and their derivatives, we see that  the system above would simplify to 
\begin{equation}
u_{14}'=\frac{\tilde{P}^{(1)}_{14}(t,u_{14},v_{14})}{v_{14} \tilde{Q}^{(1)}_{14}(t,u_{14},v_{14})},\qquad v_{14}'=\frac{\tilde{P}^{(2)}_{14}(t,u_{14},v_{14})}{v_{14} \tilde{Q}^{(2)}_{14}(t,u_{14},v_{14})}.
\end{equation}

Now, it is convenient to recalculate everything with the coefficient functions specialised according to \eqref{condssofar}, so that the condition $P^{(1)}_{14}(t,u_{14},0)=0$ becomes
\begin{equation}\label{cond14}
a_{3}''(t)=\frac{1}{2}(a_1'(t)-8 a_3(t)a_3'(t)),
\end{equation}
imposing which we find that the system becomes   quadratic-regularisable in the chart ($u_{14}$, $v_{14}$).

Similarly, with generic coefficient functions the system in the original coordinates $u_{20}$ and $v_{20}$ after the last blowup of the point $p_{20}$ is of the form
\begin{equation}
u_{20}'=\frac{P^{(1)}_{20}(t,u_{20},v_{20})}{v_{20}^2 Q^{(1)}_{20}(t,u_{20},v_{20})}, \qquad v_{20}'=\frac{P^{(2)}_{20}(t,u_{20},v_{20})}{v_{20} Q^{(2)}_{20}(t,u_{20},v_{20})}
\end{equation}
with $P^{(j)}_{20}(t,u_{20},0)$ and $ Q^{(j)}_{20}(t,u_{20},0)$ ($j=1,\,2$) nonzero and independent on $u_{20}$. Moreover, if we set  expression $P^{(1)}_{20}(t,u_{20},0)=0$, which gives some (cumbersome) relation between coefficients $a_k(t)$ and their derivatives, we see that the system above would also simplify to 
\begin{equation}
u_{20}'=\frac{\tilde{P}^{(1)}_{20}(t,u_{20},v_{20})}{v_{20} \tilde{Q}^{(1)}_{20}(t,u_{20},v_{20})}, \qquad v_{20}'=\frac{\tilde{P}^{(2)}_{20}(t,u_{20},v_{20})}{v_{20} \tilde{Q}^{(2)}_{20}(t,u_{20},v_{20})}.
\end{equation}
Again recalculating the system in this chart with the assumptions \eqref{condssofar}, we see that the condition $P^{(1)}_{20}(t,u_{20},0)=0$ becomes  
\begin{equation}\label{cond20}
a_{3}''(t)=\frac{1}{2}(8 a_3(t)a_3'(t)-a_1'(t)),
\end{equation}
imposing which causes the system to become   quadratic-regularisable in the chart ($u_{20}$, $v_{20}$).

The final step is to solve the conditions \eqref{condssofar}, \eqref{cond14} and \eqref{cond20} and show that this must recover the Takasaki system without loss of generality.
We see that $a_{-3}(t)=c$, where $c$ is a constant. 
Next, from \eqref{cond14} and \eqref{cond20} we find $a_3''(t)=0$ and $a_1'(t)=8 a_3(t)a_3'(t)$, from which we can deduce that $a_3(t)$ is linear in $t$ and $a_1(t)$ is quadratic in $t$.
If we let 
\begin{equation}
a_3(t)=c_1+c_2 t,
\end{equation} 
then 
\begin{equation}
a_1(t)=c_3+8t c_1 c_2+4t^2 c_2^2.
\end{equation} 
We also have that conditions (\ref{cond14}) and (\ref{cond20}) are satisfied.  
By changes of independent variable $t \mapsto a t + b$ we can without loss of generality set particular values of the constants $c_i$, $i=1,2,3$, as $c_1=0,$ $c_2=1/2$, $c_3=-1+a_1+2a_2$ and $c=-16 a_1^2$, which gives $a_{-3}(t)=8 \beta$, $a_3(t)=t/2$, $a_1(t)=t^2-\alpha$, so we completely recover the Takasaki system. 
This proves the theorem.
\end{proof}

The following is proved along exactly the same lines as the above proposition. Note that in comparison with \cite{HK} we do not make any Nevanlinna type estimates and the whole procedure is computational.  
\begin{proposition} \label{proprecoverHK}
If the equation
\begin{equation*}
y''=\frac{3}{4}y^5+\sum_{k=0}^3a_k(t) y^k
\end{equation*}
is quadratic regularisable on a bundle of rational surfaces obtained by blowing up points in the same configuration as $\X^{\operatorname{HK}}$, then it must coincide with the Halburd-Kecker equation up to affine changes of independent variable.

\end{proposition}
\begin{remark}
If we assume at the beginning that 
\begin{equation*}
y''=\frac{3}{4}y^5+\sum_{k=0}^4a_k(t) y^k
\end{equation*}
and quadratic regularisability, the computations and conditions are a bit more involved and include $a_4$ but it is necessary to take $a_4=0$ for quadratic regularisability. 
\end{remark}


In a similar way, the definition of $n$-th order regularisablity can also be useful to eliminate whole classes of equations that have no algebroid solutions of certain types. 
Let us consider, for instance, equation
\begin{equation} \label{ex7}
y''=y^7+\sum_{k=0}^5 a_k(t)y^k
\end{equation}
or equation 
\begin{equation}\label{ex4}
y''=y^4+\sum_{k=0}^2 a_k(t)y^k
\end{equation}
and assume that the equivalent systems are cubic regularisable (that is, the function $w=y^3$ will be meromorphic). 
After resolving all indeterminancies through blowups, we verify that we cannot cubically regularise the systems for any choice of $a_k$'s. 
We remark that this can also be seen from local expansions of the solutions of the equations.

\section{Discussion} 

To summarise, the guiding idea for this paper is the interpretation of Okamoto's spaces of initial conditions for the Painlev\'e equations as a mechanism by which equations with a special singularity structure (the Painlev\'e property) can be associated to rational surfaces via an appropriate notion of regularisation (uniform foliation of the bundle of surfaces).
The main goal of this paper was to extend this construction to a wider class of equations with a weaker restriction on their singularity structures through a more general notion of regularisation. 

Note that the way we formulate this for equations with the algebro-Painlev\'e property allows the analogy to extend to the theory of global Hamiltonian structures of Painlev\'e equations on their defining manifolds.
In the Painlev\'e case, global polynomial Hamiltonian structures with respect to symplectic forms which are nowhere zero lead to regular differential systems on the defining manifolds.
In the case of the algebro-Painlev\'e systems we consider, this is generalised to a global polynomial Hamiltonian structure with respect to a symplectic form with certain zeroes allowed - where the symplectic form is non-zero the system is automatically regular, and where it has zeroes the the Hamiltonian functions being polynomial of a certain form ensures the system is regularisable.

We must remark that the equations from the second-order non-autonomous algebro-Painlev\'e class we consider are related to Painlev\'e equations or their solvable autonomous limits via algebraic transformations, so admittedly the geometric picture is unlikely to give new insights in terms of properties of the equations, since everything must map down to the Painlev\'e case where the geometric picture is well-understood. 
However the next step is to to perform similar constructions for examples which have the algebro-Painlev\'e property but which do not reduce to Painlev\'e. 
To this end, it will be interesting to find an algorithmic procedure to  find classes of equations with globally finite branching, that is, equations satisfying both a second-order differential equation $y''=F(t,y,y')$ and an algebraic relation with meromorphic coefficients $y^n+s_1(t)y^{n-1}+\ldots+s_{n-1}(t)y+s_n(t)=0$. 
One approach is to find a differential equation for $s_n$ (of second order and of degree higher than one) and find conditions on other coefficients of both the algebraic and differential equations necessary for $s_n$ to be a meromorphic function. 
The difficulty here lies in finding an equivalent system of first order differential equations for which the method of blowing up singularities works well. 

The construction presented here also has potential applications in proving the algebro-Painlev\'e property. 
In the case of the differential Painlev\'e equations, the regularisation on the defining manifold is not sufficient to prove the Painlev\'e property - having regular initial value problems everywhere is one part, but afterwards certain auxiliary functions are required to complete the proofs. 
The geometry of the defining manifold can to an extent give clues as to the construction of appropriate auxiliary functions as in \cite{Takanodefiningmanifolds}, so it is natural to ask if methods of proofs of the algebro-Painlev\'e property could be constructed purely from manifolds on which a system is algebraically regularisable.

Finally, the most interesting continuation of the ideas of this paper is a further generalisation to account for quasi-Painlev\'e equations, whose solutions exhibit locally but not globally finite branching about movable singularities, and in particular those with Hamiltonian structures.
A notion of regularisability, and properties of a global Hamiltonian structure on a bundle of rational surfaces which guarantees it, will be presented and illustrated in a subsequent manuscript.
The ability to associate rational surfaces to equations in the quasi-Painlev\'e case is particularly exciting, since the geometry of the associated surfaces would no longer just boil down to Painlev\'e, and geometry would once again have the potential to tell one everything about the equations.

\section*{Acknowledgements} 
GF acknowledges the support of the National Science Center (Poland) via grant OPUS2017/25 /B/BST1/00931. 
GF is also partially supported by the Ministry of Science and Innovation of Spain and the European Regional Development Fund (ERDF) [grant number XXX].
Part of this work was done while AS was supported by London Mathematical Society Early Career Fellowship ECF-2021-24, and AS gratefully acknowledges the support of the London Mathematical Society.
AS was supported during the preparation of this manuscript by a Japan Society for the Promotion of Science (JSPS) Postdoctoral Fellowship for Research in Japan and also acknowledges the support of JSPS KAKENHI Grant Number 21F21775.

\begin{appendices}

\section{Surfaces associated with the Okamoto system} \label{appendixA}

\subsection{Blowups} \label{appendixA1}

We begin with $(f,g)$ as coordinates on $\C^2$, which we compactify to $\p^1 \times \p^1$ by introducing $F = 1/f, G= 1/g$, so $\p^1 \times \p^1$ is covered by the four affine charts $(f,g), (F,g), (f,G), (F,G)$.
In performing the sequence of eight blowups to obtain $\X^{\Ok}$ we use the following convention for introducing charts.
After blowing up a point $p_i $ given in some affine chart $(x,y)$ by
 $(x,y) = (x^*, y^*)$, the exceptional divisor $E_i \cong \mathbb{P}^1$ replacing $p_i$ can be covered by two local affine coordinate charts $(u_i,v_i)$ and $(U_i,V_i)$ given by
\begin{equation} \label{blowupchartconvention}
\begin{aligned}
x- x^{*} &= u_i v_i ,&\quad &y-y^{*} = v_i,\\
x- x^{*} &= V_i ,&\quad &y-y^{*} = U_i V_i.
\end{aligned}
\end{equation}
In particular the exceptional divisor $E_i$ has in these charts local equations $v_i=0$, respectively $V_i=0$. Below we present the locations in coordinates of the points to be blown up to obtain $\X^{\Ok}$ as in \autoref{fig:surfaceokamoto}, with arrows $p_{i-1} \leftarrow p_i$ indicating that $p_i$ lies on the exceptional divisor arising from the blowup of $p_{i-1}$. We use the convention \eqref{blowupchartconvention} for charts, so e.g. $p_1 : (F,g) = (0,0)  \leftarrow p_{2} : (U_1, V_1) = ( -a_2,0)$ indicates that the coordinates $(U_1,V_1)$ in which we find $p_2$ are defined by $F=V_1$, $g = U_1 V_1$. 
\begin{equation*}
\begin{aligned}
&p_1 : (F,g) = (0,0)  &\leftarrow     &\quad p_{2} : (U_1, V_1) = ( -a_2,0) \\
&p_3: (f,G) = (0, 0)  &\leftarrow     &\quad p_{4} : (u_3, v_3) = ( a_1 ,0)\\
&p_5 : (F,G) = (0,0) &\leftarrow     &\quad p_{6} : (U_5, V_5) = (2,0) \\
&   & &\quad   \uparrow  & \\
&   & &\quad  p_{7} : (u_{6},v_{6})  = (-4 t, 0)  \\
&   & &\quad   \uparrow  & \\
&   & &\quad  p_{8} : (u_{7},v_{7}) = \left( 4(1 - a_1 - a_2 + 2 t^2) , 0\right)  \\
\end{aligned}
\end{equation*}
The inaccessible divisors whose support $D^{\Ok}$ is removed from $\X^{\Ok}$ as part of the construction of Okamoto's space are given by
\begin{equation}
\begin{aligned}
D_0 &= E_3 - E_4,\\
D_1 &= E_1 - E_2, \\
D_2 &= H_f - E_1 - E_5, \\
D_3 &= E_5 - E_6,
\end{aligned}
\qquad 
\begin{aligned}
D_4 &= E_6 - E_7, \\
D_5 &= E_7 - E_8, \\
D_6 &= H_g - E_3 - E_5, 
\end{aligned}
\end{equation}
where $H_f$, $H_g$ denote the total transforms of lines of constant $f$, $g$, respectively on $\p^1 \times\p^1$, so for example $D_2 = H_f - E_1 - E_5$ is the proper transform of the line $F = 0$. 
In particular 
\begin{equation}
D^{\Ok} = \bigcup_{i=0}^6 D_i,
\end{equation}
Note that we have used a slightly different enumeration to that in \cite{KNY}, which we have chosen to make the comparison of intersection graphs of inaccessible divisors on $\X^{\Tak}_m$, $\X^{\Ok}$ neater.

\subsection{Symplectic atlas}

While the charts introduced in the canonical way in the previous section are sufficient to provide an atlas for $E^{\Ok}$ to which the system \eqref{hamOkP4} extends to be everywhere regular, to obtain the atlas due to Matano, Matumiya and Takano \cite{MMT99} (see also \cite{ST97}) in which all Hamiltonians are polynomial, one additional change of coordinates is required.

The coordinates introduced after the blowups of $p_1, \dots, p_4$ remain unchanged from above, but after the blowup of $p_5$ one must make the local change of coordinates
\begin{equation}
U_5 \mapsto \tilde{U}_5 = \frac{1}{U_5}, \quad V_5 \mapsto \tilde{V}_5 = V_5,
\end{equation}
which corresponds to, after the blowups of $p_1,\dots, p_5$, changing the affine coordinate $U_5$ parametrising the exceptional divisor $E_5 \cong \p^1$. 
Implementing this, the rest of the charts are introduced in the canonical way as in the previous subsection, but relabeled $\tilde{u}_i, \tilde{v}_i$ to distinguish them from those above. 
These are given by the following:
\begin{equation*}
\begin{aligned}
&p_5 : (F,G) = (0,0) &\leftarrow     &\quad p_{6} : (\tilde{U}_5, \tilde{V}_5) = (1/2,0) \\
&   & &\quad   \uparrow  & \\
&   & &\quad  p_{7} : (\tilde{u}_{6},\tilde{v}_{6})  = (t, 0)  \\
&   & &\quad   \uparrow  & \\
&   & &\quad  p_{8} : (\tilde{u}_{7},\tilde{v}_{7}) = \left( a_1 + a_2 -1 , 0\right)  \\
\end{aligned}
\end{equation*}
After this, we obtain the charts from the atlas \eqref{atlasOk} by letting 
\begin{equation}
(x_1,y_1) = (u_2,v_2), \qquad (x_2,y_2) = (v_4,u_4), \qquad (x_3,y_3) = (\tilde{u}_8, \tilde{v}_8), 
\end{equation} 
which are related to the coordinates $(f,g)$ by the gluing \eqref{gluingOk}. The fact that these provide canonical coordinates for $\omega^{\Ok}$ as in \eqref{canonicalcoordsOk} is verified by direct calculation.
\subsection{Extended Okamoto surfaces}
The locations of the extra points $z_1,\dots z_8$ on $\X^{\Ok}$ to be blown up to obtain $\tilde{\X}^{\Ok}$ as indicated in \autoref{fig:extendedOkamotosurface} are given by
\begin{equation}
\begin{aligned}
z_1 : (u_3, v_3) &= (0,0), &&\quad  z_1 \in  
D_0 \cap  \left(H_f - E_3 \right),\\
z_2 : (U_4, V_4) &= (0,0), &&\quad  z_2 \in 
D_0 \cap E_4, \\
z_3 : (U_2, V_2) &= (0,0), &&\quad  z_3 \in 
D_1 \cap D_2, \\
z_4 : (u_1, v_1) &= (0,0), &&\quad  z_4 \in 
D_1 \cap E_2, \\
z_5 : (u_5, v_5) &= (0,0), &&\quad  z_5 \in 
D_3 \cap D_2, \\
z_6 : (U_6, V_6) &= (0,0), &&\quad  z_6 \in 
D_3 \cap D_4, \\
z_7 : (U_7, V_7) &= (0,0), &&\quad  z_7 \in 
D_4 \cap D_5, \\ 
z_8 : (U_8, V_8) &= (0,0), &&\quad  z_8 \in 
D_5 \cap E_8,
\end{aligned}
\end{equation}
where we the coordinates $(u_i,v_i)$, $(U_i, V_i)$ are those introduced using the point locations in \autoref{appendixA1}, so without the coordinate change on $E_5$ which was used to construct the symplectic atlas.

For computations on $\tilde{\X}^{\Ok}$ we use charts $(r_i,s_i)$ and $(R_i, S_i)$ to cover the exceptional divisor $L_i = \rho^{-1}(z_i)$, defined according to the same convention as $(u_i,v_i)$ and $(U_i,V_i)$, namely
\begin{equation}
\begin{aligned}
u_3 &= r_1 s_1, 	 &v_3 = s_1, 	&&\quad &&u_3 = S_1, 	&&v_3 = R_1 S_1, \\ 
U_4 &= r_2 s_2, 	 &V_4 = s_2, 	&&\quad &&U_4 = S_2, 	&&V_4 = R_2 S_2, \\
U_2 &= r_3 s_3, 	 &V_2 = s_3, 	&&\quad &&U_2 = S_3, 	&&V_2 = R_3 S_3, \\
u_1 &= r_4 s_4, 	 &v_1 = s_4, 	&&\quad &&u_1 = S_4, 	&&v_1 = R_4 S_4, \\
u_5 &= r_5 s_5, 	 &v_5 = s_5, 	&&\quad &&u_5 = S_5, 	&&v_5 = R_5 S_5, \\
U_6 &= r_6 s_6, 	 &V_6 = s_6, 	&&\quad &&U_6 = S_6, 	&&V_6 = R_6 S_6, \\
U_7 &= r_7 s_7, 	 &V_7 = s_7, 	&&\quad &&U_7 = S_7, 	&&V_7 = R_7 S_7, \\
U_8 &= r_8 s_8, 	 &V_8 = s_8, 	&&\quad &&U_8 = S_8, 	&&V_8 = R_8 S_8,
\end{aligned}
\end{equation}
so that $L_i$ has local equation $s_i=0$, respectively $S_i=0$.

\section{Surfaces associated with the Takasaki system}  \label{appendixB}

\subsection{Blowups}

Just as in the Okamoto case we begin with $(q,p)$ as coordinates on $\C^2$, which we compactify to $\p^1 \times \p^1$ by introducing $Q = 1/q, P= 1/p$, so $\p^1 \times \p^1$ is covered by the four affine charts $(q,p), (Q,p), (q,P), (Q,P)$.
We now present the locations of points $q_1,\dots, q_{20}$ as depicted in \autoref{fig:takasakisurface}, using precisely the same convention as in \autoref{appendixA} for introduction of charts after a blowup, recycling the notation $(u_i,v_i)$, $(U_i,V_i)$.
We first blowup up $q_1$, then find two distinct point $q_2$ and $q_4$ on the exceptional divisor $F_1$. 
Two blowups are then performed over each of these, the centres of which are given in coordinates as follows.
\begin{equation*}
\begin{gathered}
q_1 : (q,P) = (0,0) \quad \leftarrow \quad
\left\{
\begin{aligned}
q_2 : (u_1,v_1) &= (4 a_1, 0) ~&\leftarrow &\quad q_3 : (u_2,v_2) = (0,0) \\
q_4: (u_1,v_1) &= (-4a_1, 0) ~&\leftarrow &\quad q_5 : (u_4, v_4) = (0,0)
\end{aligned}
\right.
\end{gathered}
\end{equation*}
After four repeated blowups over $q_6$ we again find two distinct point on the exceptional divisor $F_8$:
\begin{equation*}
\begin{gathered}
q_6 : (Q,P) = (0,0), ~ \leftarrow ~ q_7 : (U_7,V_7)= (0,0) ~ \leftarrow ~ q_8 : (u_7,v_7) = (0,0) ~ \leftarrow ~
\left\{
\begin{aligned}
&q_9  \\
&q_{15}
\end{aligned}
\right.
\end{gathered}
\end{equation*}
We perform six blowups over each of these, of points given in coordinates over $q_9$ as given by the following:
\begin{equation*}
\begin{aligned}
&q_9 : (u_8,v_8) = (8,0) ~ \leftarrow  &~   &q_{10} : (u_9, v_9) = (0,0) \\
&  & &  \uparrow  & \\
&  & & q_{11} : (u_{10},v_{10})  = (-64 t, 0)  \\
&  & &  \uparrow  & \\
&  & & q_{12} : (u_{11},v_{11}) = (0, 0)  \\
&  & &  \uparrow  & \\
&  & & q_{13} : (u_{12},v_{12}) = ( 256(2 + 2 t^2 - a_1 -2 a_2), 0)  \\
&  & &  \uparrow  & \\
&  & & q_{14} : (u_{13},v_{13}) = (0, 0) 
\end{aligned}
\end{equation*}
Similarly over $q_{15}$ we perform six blowups:
\begin{equation*}
\begin{aligned}
&q_{15} : (u_8,v_8) = (-8,0) ~ \leftarrow  &~   &q_{16} : (u_{15}, v_{15}) = (0,0) \\
&  & &  \uparrow  & \\
&  & & q_{17} : (u_{16},v_{16})  = (64 t, 0)  \\
&  & &  \uparrow  & \\
&  & & q_{18} : (u_{17},v_{17}) = (0, 0)  \\
&  & &  \uparrow  & \\
&  & & q_{19} : (u_{18},v_{18}) = ( - 256(2 t^2 - a_1 -2 a_2), 0)  \\
&  & &  \uparrow  & \\
&  & & q_{20} : (u_{19},v_{19}) = (0, 0) 
\end{aligned}
\end{equation*}
We note that in the above point locations, there are several instances of points lying on the proper transforms of exceptional divisors from prior blowups. 
These are as follows, where we use $\operatorname{Bl}_{q_i \dots q_j}$ to denote the projection from the blowups of $q_1, \dots, q_j$.
\begin{itemize}
\item $q_7$ lies on the proper transform under $\operatorname{Bl}_{q_6}$ of the line $\{P=0\}$.
\item $q_8$ lies on the proper transform under $\operatorname{Bl}_{q_6 q_7}$ of the line $\{P=0\}$.
\end{itemize}
To describe the inaccessible divisors on the resulting surface $\X^{\Tak}$ we let $H_q$, $H_p$ be divisors giving the total transforms of lines of constant $q$, $p$ on respectively.
The inaccessible divisors are 
\begin{equation} \label{inaccdivcompsTak}
\begin{gathered}
H_q - F_1, \quad F_1 - F_2 - F_4, \quad F_2 - F_3, \quad F_4 - F_5, \\
H_p - F_1 - F_6 - F_7 - F_8, \quad F_8 - F_9 - F_{15}, \quad F_7 - F_8, \quad F_6 - F_7, \quad H_q - F_6, \\
F_9 - F_{10}, \quad F_{10} - F_{11},  \quad F_{11} - F_{12}, \quad F_{12} - F_{13}, \quad F_{13} - F_{14}, \\
F_{15} - F_{16},  \quad F_{16} - F_{17}, \quad F_{17} - F_{18}, \quad F_{18} - F_{19},  \quad F_{19} - F_{20},
\end{gathered}
\end{equation}
for which we do not introduce labels since they will change under the minimisation in the next subsection.

\subsection{Minimisation}

We note that among the components of $D^{\Tak}$ as listed in \eqref{inaccdivcompsTak}, there are two exceptional curves of the first kind, i.e. those of self-intersection $-1$. These are $H_q - F_1$, which is the proper transform of the line $\{ q = 0\}$, and $H_q - F_6$, which is the proper transform of the line $\{ Q = 0 \}$. 
Under the birational morphism contracting both of these, we note that another component $F_6 - F_7$ is sent to an exceptional curve of the first kind which we can then contract. 
Blowing this down, we again see another component becoming of self-intersection $-1$, namely $F_7 - F_8$.
Under the projection 
\begin{equation}
\sigma : \X^{\Tak} \rightarrow \X^{\Tak}_m ,
\end{equation} 
from this sequence of four blowdowns, we find no more components of $\sigma(D^{\Tak})$ which are contractable through blowdowns. 
This leads to the curves $C_i$, $i=1, \dots, 15$, as shown on \autoref{mintakasakisurface}, given explicitly as divisors on $\X^{\Tak}_m$ by
\begin{equation}
\begin{aligned}
C_1 &= F_2 - F_3, 					&& C_6 	= F_9 - F_{10}, 			&&& C_{11} 	= F_{15} - F_{16}, \\
C_2 &= F_4 - F_5,					&& C_7 	= F_{10} - F_{11},		&&& C_{12} 	= F_{16} - F_{17},  \\
C_3 &= \sigma(F_1 - F_2 - F_4), 		&& C_8 	= F_{11} - F_{12},		&&& C_{13} 	= F_{17} - F_{18}, \\
C_4 &= H_p - F_1 - F_6 - F_7 - F_8,		&& C_{9}	= F_{12} - F_{13}, 		&&& C_{14}	= F_{18} - F_{19},\\
C_5 &= \sigma(F_8 - F_9 - F_{14})		&& C_{10}	 = F_{13} - F_{14},		&&& C_{15} 	= F_{19} - F_{20},
\end{aligned}
\end{equation} 
where we write the image $\sigma(D)$ simply as $D$ if the divisor $D$ is disjoint from the curves that are blown down by $\sigma$, since in this case $D$ and $\sigma(D)$ are isomorphic.

\end{appendices}

\end{document}